\documentclass[12pt]{article}
\usepackage{amsmath}
\usepackage{amssymb}
\usepackage{amsthm}
\usepackage{graphicx,psfrag,epsf}
\usepackage{psfrag,epsf}
\usepackage{enumitem}
\usepackage{natbib}
\usepackage{url} 
\usepackage[titletoc,title]{appendix}
\usepackage{array}
\usepackage{bm}
\usepackage[dvipsnames]{xcolor}
\usepackage{pifont}
\usepackage{color}
\usepackage{graphicx}
\usepackage{epstopdf}
\usepackage[letterpaper, portrait, margin=1in]{geometry}
\usepackage{lipsum}
\usepackage{listings}
\usepackage{fancyref}
\usepackage{float}
\usepackage[font=footnotesize,labelfont=bf]{caption}
\usepackage{color}
\usepackage{xr}
 \usepackage{caption}
 \usepackage{subcaption}
 
\usepackage{color}

\newtheorem{theorem}{Theorem}[section]

\newtheorem{corollary}{Corollary}[section]

\theoremstyle{remark}
\newtheorem*{remark}{Remark}

\newcommand*\taufs{\tau_{\mathcal{S}}}
\newcommand*\taukfs{\tau_{k, \mathcal{S}}}
\newcommand*\tauk{\tau_{k}}

\newcommand*\poobs{Y_{i}^{obs}}
\newcommand*\pot{Y_{i}(t)}
\newcommand*\poc{Y_{i}(c)}
\newcommand*\poz{Y_{i}(z)}
\newcommand*\meanpot{\bar{Y}(t)}
\newcommand*\meanpoc{\bar{Y}(c)}
\newcommand*\meanpoz{\bar{Y}(z)}
\newcommand*\meanpotk{\bar{Y}_k(t)}
\newcommand*\meanpock{\bar{Y}_k(c)}
\newcommand*\meanpozk{\bar{Y}_k(z)}

\newcommand*\meanpozobs{\bar{Y}^{obs}(z)}

\newcommand*\meanpozkobs{\bar{Y}_k^{obs}(z)}
\newcommand*\popmeanpoz{\mu(z)}
\newcommand*\popmeanpot{\mu(t)}
\newcommand*\popmeanpoc{\mu(c)}
\newcommand*\popmeanpozk{\mu_k(z)}
\newcommand*\popmeanpotk{\mu_k(t)}
\newcommand*\popmeanpock{\mu_k(c)}

\newcommand*\model{\mathcal{F}}

\newcommand*\EE{\mathbb{E}}

\newcommand*\nk{n_k}
\newcommand*\nck{n_{c, k}}
\newcommand*\ntk{n_{t, k}}
\newcommand*\nzk{n_{z, k}}

\newcommand*\tauestbk{\widehat{\tau}_{(BK)}}
\newcommand*\tauestcr{\widehat{\tau}_{(CR)}}
\newcommand*\taukest{\widehat{\tau}_{k}}

\newcommand*\Sc{S^{2}(c)}
\newcommand*\St{S^{2}(t)}
\newcommand*\Sz{S^2(z)}
\newcommand*\Stc{S^{2}(tc)}
\newcommand*\scest{s^{2}(c)}
\newcommand*\stest{s^{2}(t)}
\newcommand*\szest{s^{2}(z)}

\newcommand*\Sck{S^{2}_k(c)}
\newcommand*\Stk{S^{2}_k(t)}
\newcommand*\Szk{S^{2}_k(z)}
\newcommand*\Stck{S^{2}_k(tc)}
\newcommand*\sckest{s^{2}_k(c)}
\newcommand*\stkest{s^{2}_k(t)}
\newcommand*\szkest{s^{2}_k(z)}

\newcommand*\sigmac{\sigma^{2}(c)}
\newcommand*\sigmat{\sigma^{2}(t)}
\newcommand*\sigmaz{\sigma^{2}(z)}
\newcommand*\sigmatc{\sigma^{2}(tc)}
\newcommand*\sigmack{\sigma^{2}_k(c)}
\newcommand*\sigmatk{\sigma^{2}_k(t)}
\newcommand*\sigmazk{\sigma^{2}_k(z)}
\newcommand*\sigmatck{\sigma^{2}_k(tc)}

\newcommand*\varestcr{\widehat{\sigma}^{2}_{(CR)}}
\newcommand*\varestk{\widehat{\sigma}^{2}_{k}}
\newcommand*\varestbkone{\widehat{\sigma}^{2}_{(BK)}}

\begin{document}
\def\spacingset#1{\renewcommand{\baselinestretch}%
{#1}\small\normalsize} \spacingset{1}


  \title{Block what you can, except when you shouldn’t\thanks{\noindent{\textit{Email: } \texttt{nicole.pashley@rutgers.edu}. The authors would like to thank Guillaume Basse, Avi Feller, Colin Fogarty, Michael Higgins, Luke Keele, and Lo-Hua Yuan for their comments and edits on an earlier version of this paper.
We would also like to thank members of Luke Miratrix's and Donald B. Rubin's research labs for their useful feedback on the project and Peter Schochet and Kosuke Imai for insightful discussion of this material.
Special thanks to Joan Heller of Heller Research Associates for access to the data for the numerical example.  Thanks to anonymous reviewers for many comments, including ideas regarding the justifications of blocking in practice.
The research reported here was partially supported by the Institute of Education Sciences, U.S. Department of Education, through Grant R305D150040.
This material is also based upon work supported by the National Science Foundation Graduate Research Fellowship under Grant No. DGE1745303.
Any opinion, findings, and conclusions or recommendations expressed in this material are those of the authors and do not necessarily reflect the views of the National Science Foundation, the Institute of Education Sciences or the U.S. Department of Education.
}}}
    \author{Nicole E. Pashley\\Department of Statistics, Rutgers University  \and Luke W. Miratrix\\Graduate School of Education, Harvard University}
 \maketitle

\bigskip
\begin{abstract}
Several branches of the potential outcome causal inference literature have discussed the merits of blocking versus complete randomization.
Some have concluded it can never hurt the precision of estimates, and some have concluded it can hurt.
In this paper, we reconcile these apparently conflicting views, give a more thorough discussion of what guarantees no harm, and discuss how other aspects of a blocked design can cost, all in terms of precision.
We discuss how the different findings are due to different sampling models and assumptions of how the blocks were formed.
We also connect these ideas to common misconceptions, for instance showing that analyzing a blocked experiment as if it were completely randomized, a seemingly conservative method, can actually backfire in some cases.
Overall, we find that blocking can have a price, but that this price is usually small and the potential for gain can be large.
It is hard to go too far wrong with blocking.
\end{abstract}

\noindent%
\vfill

\section{Introduction}
\spacingset{1.45} 

Many of us may have embraced George Box's famous quote\footnote{Typically the quote is attributed to Box and appears on page 103 of \cite{box1978statistics}.} ever since it was thrown at us during an undergraduate statistics course:
\begin{center}
``Block what you can and randomize what you cannot.''
\end{center}
``Blocking'' an experiment---the act of grouping similar units together and then randomizing them into treatment and control within each group---is a tool for increasing the precision of an estimate of a treatment impact.
Blocking may also be a consequence of how the units were obtained in the first place.
Many of us interpret this fragment of Box's advice\footnote{Box actually provided more nuanced advice regarding the advantages and disadvantages of blocking than is shown in this popular quote.} to suggest that we should block on whatever characteristics and information we have.
But is blocking always ``worth it''?
Could it ever be a mistake, causing more harm than good?
And if so, when?
In this paper we unpack these questions, showing that their answers depend on details often left implicit when thinking about blocking, such as whether the blocks are formed by the researcher or are inherent to the context, or how the experimental sample is obtained.
We additionally shed light on some common misconceptions regarding comparisons of blocked designs and completely randomized (i.e., unblocked) designs.

While blocking has been investigated from many perspectives \citep[e.g.][]{snedecor1989statistical}, we focus on the causal inference potential outcomes framework.
In this literature there is, on the face of it, disagreement regarding the guarantees of blocking.
\cite{Imai2008} compares the true variance for the matched pairs design (i.e., blocked experiments where all blocks have two units) to complete randomization with random sampling of pairs for both designs and with random sampling of pairs for the matched design but simple random sampling for complete randomization.
That paper concludes ``the relative efficiency of the matched-pair design depends on whether matching induces positive or negative correlations regarding potential outcomes within each pair'' \citep[][p. 4865]{Imai2008}.
A similar comment is made on  p. 101 of \cite{snedecor1989statistical}, with both implying that the matched pairs design may be helpful or harmful. 
In contrast, \cite{imbens2011experimental}, assuming a stratified sampling superpopulation model, claims that ``In experiments with randomization at the unit-level, stratification is superior to complete randomization, and pairing is superior to stratification in terms of precision of estimation'' (p. 1).
\cite{imai_king_stuart_2008} similarly concludes that the variance under the blocked design (with at least two units assigned to treatment and control within each block) is never higher than under complete randomization for a superpopulation set-up.
Researchers are in more agreement regarding finite sample inference, for which blocking can be helpful or harmful.
All of these conclusions are correct: their differences stem from differences in the specific sampling frameworks used for the respective analyses.
While considerations of the finite vs. superpopulation distinction or block size are important, they are not enough: how the blocks are formed and the specific details of any sampling framework being used also matter.
This is the work of this paper.

In particular, we consider multiple block types and sampling mechanisms.
For each combination of these elements considered, we derive expressions comparing the variance of estimates under blocking to those under complete randomization.
This allows us to show which overall frameworks guarantee no harm from blocking and which do not.
We focus on comparing true variance, rather than delving into the problem of variance estimation \citep[see][for discussion of variance estimation]{pashley_main}.
For the superpopulation contexts, we also carefully separate out variance due to the sampling of units from variance due to randomized treatment assignment; this gives more precise statements of the benefits of blocking than have been given in prior literature.
We generally follow the taxonomy of block types and sampling mechanisms discussed in \citet{pashley_main}.
In particular, certain sampling mechanisms go hand in hand with certain types of blocks, where the type of block captures how the blocks are formed.
These frameworks and associated types of blocks are as follows:
\begin{enumerate}
	\item Finite sample: This is the case where the units in the experiment are considered fixed and the only source of randomness is the treatment assignment itself. Here blocks may be formed in any way pre-randomization using measured baseline covariates, and we consider them fixed with the sample. Consider a psychology researcher evaluating an intervention on a convenience sample of students grouped by baseline ability.
	\item Simple random sampling: Here we would sample units from a larger population and then divide the units into blocks based on some covariate(s). These \emph{flexible} blocks are a consequence of the units we have, such that before sampling we may not even know how many blocks we will have or what size they will be. If our researcher from above viewed their experimental sample of units as randomly drawn from some target population, we would use this setting.
	\item Stratified random sampling: Here we sample a specified number of units from each of a pre-specified set of strata, and then randomize within these groups. Consider a researcher recruiting a sample of children from each of a series of ages, and blocking by age. The blocks are inherently due to \emph{fixed} aspects of the units themselves, which also define the population strata. This setting is noteworthy as it is often the implicit assumption made when comparing blocking to complete randomization.
	\item Random sampling of strata: In this framework, the population is made up of an infinite number of blocks and entire blocks are sampled as units. This could be, e.g., a random sample of twins. The blocks themselves are \emph{structural} in that they are naturally grouped as a product of the world, not the researcher. This setting is often assumed when comparing the matched pairs design to complete randomization.
	\item Two stage sampling: In this framework, we first select a sample of strata from an infinite population of strata, as just above, but then, within each sampled stratum, draw a sample of individuals. We consider the selected strata to be of infinite size, making the sampled individuals within each block a sample drawn from a stratum-specific superpopulation.  These superpopulations are not fixed as they are in the stratified random sampling case due to the sampling of strata. Consider a \emph{multisite trial} where a sample of students is obtained for each of a random sample of schools.
\end{enumerate}
In Section~\ref{sec:bk_vs_cr} we, for each of these contexts, carefully compare blocking vs. complete randomization and provide closed form formulae for how the variances of an estimator under these approaches will differ.
We also use these formulae to provide guidance on using blocking across different contexts.
We provide the derivations for these formulae in Supplementary Material~\ref{append:cr_vs_bk}.

Once we finish these comparisons we broaden our scope to consider how blocking is commonly done in practice.
Our primary investigation, following prior literature on this tension, only pertains to the rare case when all blocks manage to have the exact same proportion of units treated.
This is not often the case.
If we do not assume equal proportions across blocks, then units in the same treatment arm can be weighted differently, which can reduce the efficiency of blocking regardless of framework.
We discuss this cost further in Section~\ref{sec:cr_vs_bk_unequal}.

Next, in Section~\ref{sec:ignore_bk}, we address two common misconceptions regarding blocking.
We first examine the question of whether analyzing a blocked randomized experiment as if it were completely randomized is in fact a safe approach.
It is not: for the same reasons that implementing blocking could potentially have a cost, ignoring it can as well.
We then discuss whether the variance estimator for a completely randomized experiment is more stable than that of a blocked experiment on the same data.
This is not necessarily the case: the variance of the variance estimator of a completely randomized experiment can either be more, or less, stable than the variance estimator for a blocked experiment.

Finally, in Section~\ref{sec:sims}, we move away from the theoretical discussion and present a few illustrative simulations.
We first present a range of scenarios from blocking being mildly costly to blocking being very advantageous to underscore how even slight success in the grouping of units easily offsets blocking's cost.
We also show that, while a researcher can deliberately form blocks to reduce precision, it is difficult, but not impossible, for a researcher to unintentionally form blocks from a random sample of units in a way that is disadvantageous.

Throughout our paper, we show that while claims of ``no harm'' are often unfounded, the potential harm is generally going to be minimal.
Our primary aim is to unpack the tensions at play in order to help guide further work in the field, and help to lay to rest this apparent debate.
In the end, we advise researchers to not be afraid to block, but to not bother blocking on things that are unlikely to be related to one's outcomes.
We believe our findings will generalize to related designs intended to reduce imbalance on covariates between the treatment and control groups, such as rerandomization \citep{morgan2012rerandomization, morgan2015rerandomization, branson2016improving, li2018asymptotic, schultzberg2019re}, but we do not explore that connection here. 

\section{Set Up} \label{sec:setup}
Before delving into comparisons of blocking and complete randomization, we first review these two experimental designs, define our notation, and review some standard, well-known, results.

\subsection{Experimental designs and estimands}
We use the Neyman-Rubin model of potential outcomes (Rubin, \citeyear{rubin_1974}; Splawa-Neyman et al., 1923/\citeyear{neyman_1923}), and assume the Stable Unit Treatment Value Assumption, i.e. no interference and no multiple forms of treatment \citep[][]{rubin_1980}.
In this framework, each unit has a potential outcome under treatment and a potential outcome under control, denoted $\pot$ and $\poc$, respectively, for unit $i$ ($i = 1,\dots,n$).
We consider two experimental designs on a sample of $n$ units: complete randomization and blocked randomization.
Under a completely randomized design $n_{t} = np$ of the units are randomly assigned to treatment, with the rest of the $n_{c} = n- n_{t}$ units assigned to control, for fixed $p \in (0,1)$.
Under a blocked design, the units are split into $K$ blocks in some manner \citep[see][for longer discussion of block types]{pashley_main}, with $\nk$ units in the $k$th block ($k = 1,\dots,K$).
Within each block a completely randomized experiment is performed independently of other blocks, such that in the $k$th block $\ntk = p_k\nk$ are randomly assigned to treatment for fixed $p_k \in (0,1)$.
For most of the paper, except where noted, we assume $p_k = p$ for $k = 1,\dots,K$.

In addition to having two experimental designs, we also consider finite sample inference and superpopulation inference.
In finite sample inference, the units in the sample and their potential outcomes are considered fixed and randomness comes solely from the treatment assignment mechanism.
The estimand for the finite sample is then the sample average treatment effect (SATE) defined as (see, e.g., \cite{CausalInferenceText}, p. 86)
\[\taufs = \frac{1}{n} \sum_{i=1}^{n}\big(\pot - \poc\big).\]
Under blocking, we can define the SATE within block $k$, for $k=1,...,K$, as
\[\taukfs = \frac{1}{\nk} \sum_{i: b_{i} = k}\big(\pot - \poc\big),\]
 where $b_i \in \{1,\dots,K\}$ indicates the block that unit $i$ belongs to.

In the superpopulation setting, we wish to make inference for some (infinite) superpopulation rather than just the units in our experiment.
We therefore need to consider the randomness induced by sampling units from the population into the sample.
We thus have two sources of randomness: the sampling and the treatment assignment mechanism.
We write our estimand in the superpopulation setting, the population average treatment effect (PATE), as
\[\tau = \EE[\pot - \poc].\]
The unconditioned expectation denotes a direct average for all units in the superpopulation.
We can similarly define the PATE within block $k$ as
\begin{align*}
\tauk &=\EE[\pot - \poc|b_{i} = k].
\end{align*}

\subsection{Estimation and variance}\label{subsec:est_var}
There is a different standard treatment effect estimator for complete randomization and blocked randomization.
However, the estimators are the same for each design whether we are interested in the SATE or PATE (at least in the settings considered here).
Let $Z_{i} = t$ if unit $i$ is assigned treatment and $Z_{i} = c$ if unit $i$ is assigned to control, so that $\poobs = Y_i(Z_i)$ is the outcome we observe for unit $i$ given a specific treatment $Z_i$.
For complete randomization, the estimator is just the simple difference in means between treatment and control units,
\[\tauestcr= \frac{1}{n_{t}}\sum_{i=1}^{n} \mathbb{I}_{Z_i = t} \pot  - \frac{1}{n_{c}} \sum_{i=1}^{n} (1-\mathbb{I}_{Z_i = t}) \poc,\]
where $\mathbb{I}_{Z_i = t}$ is the indicator that unit $i$ received treatment.
For blocked randomization, we can define this estimator within each block as 
\[\taukest = \frac{1}{\ntk}\sum_{i: b_{i}=k} \mathbb{I}_{Z_i = t} \pot  - \frac{1}{\nck} \sum_{i: b_{i} =k} (1-\mathbb{I}_{Z_i = t}) \poc,\]
$k=1,...,K$.
Then the overall blocked randomization estimator is a weighted average of these simple difference estimators for each block,
\[\tauestbk = \sum_{k=1}^K\frac{\nk}{n}\taukest. \]

We will often take the expectation over the randomization of units to treatment for a fixed sample.
In particular, we write the expectation of estimator $\widehat{\tau}$ for a given finite sample $\mathcal{S}$ and for some assignment mechanism $\textbf{P}$, which may be complete randomization or blocked randomization, as $\EE\left[ \widehat{\tau} | \mathcal{S}, \textbf{P}\right]$.
To reduce clutter, we drop the \textbf{P} and simply write $\EE\left[ \widehat{\tau} | \mathcal{S}\right]$ if the estimator makes the assignment mechanism clear.
For superpopulation settings, we condition on the sampling mechanism, block type, and assignment mechanism, e.g.  $\EE\left[ \widehat{\tau} | \model \right]$, where $\model$ is some framework

Both treatment effect estimators are generally unbiased (save for random sampling of strata and two stage sampling as discussed in the remark below), with respect to their appropriate design, under the finite sample and superpopulation frameworks considered here.
So the difference in performance between these two experimental designs and associated estimators comes down to the differences in variance.

To discuss the true variances of our estimators $\tauestcr$ and $\tauestbk$, we need some additional notation \citep[we will follow conventions found in][]{CausalInferenceText}.
The sample variance of potential outcomes under treatment $z$ is
\[\Sz = \frac{1}{n-1} \sum_{i = 1}^{n}(\poz - \meanpoz)^{2},\]
where
\[\meanpoz = \frac{1}{n}\sum_{i=1}^{n}\poz\]
is the mean of the potential outcomes for the units in the sample under treatment $z$.
The sample variance of the individual level treatment effects is
\[\Stc = \frac{1}{n - 1}\sum_{i=1}^{n}\Big(\pot - \poc - \taufs\Big)^{2}.\]
$\meanpozk$, $\Szk$, and $\Stck$ are defined analogously over the units in block $k$.

For the finite sample, the variance of the completely randomized estimator is known to be (Splawa-Neyman et al., 1923/\citeyear{neyman_1923})
\[\text{var}(\tauestcr|\mathcal{S}) = \frac{\St}{n_{t}} + \frac{\Sc}{n_{c}} - \frac{\Stc}{n}.\]
We can use this expression within each block to get block level variances, 
\begin{equation*}
\text{var}(\taukest|\mathcal{S}) = \frac{\Stk}{\ntk} + \frac{\Sck}{\nck} - \frac{\Stck}{\nk}. 
\end{equation*}
Summing these across the independent blocks, with the weights for block sizes, gives an overall variance of
\begin{align*}
\text{var}\left(\tauestbk|\mathcal{S}\right)& = \sum_{k=1}^{K}\frac{\nk^{2}}{n^{2}}\text{var}(\taukest |\mathcal{S})= \sum_{k=1}^{K}\frac{\nk^2}{n^2}\Big(\frac{\Stk}{\ntk} + \frac{\Sck}{\nck} - \frac{\Stck}{\nk}\Big).
\end{align*}

For the superpopulation, we can use a variance decomposition to obtain, under superpopulation setting $\model$,
\begin{align*}
\text{var}\left(\widehat{\tau}|\model\right) = \EE\left[\text{var}\left(\widehat{\tau}|\mathcal{S}\right)|\model\right] +  \text{var}\left(\EE\left[\widehat{\tau}|\mathcal{S}\right]|\model\right).
\end{align*}
We provide simplified formulae in the following sections when possible to do so.

\begin{remark}
It is worth noting that under the frameworks for random sampling of strata and two stage sampling, our treatment effect estimators, for both blocking and complete randomization, are not actually unbiased for $\tau$.
This bias is induced by the random number of units in the sample creating a random denominator in our estimator.
Under the conditions given in \cite{pashley_main} to obtain variance estimators in this setting, in particular either conditioning on the block sizes or assuming that block sizes are independent of block treatment effects, the estimators are unbiased for the PATE.
With respect to comparing variances, the bias is the same under either design (which can be shown by a simple application of the law of total expectation).
Thus, these comparisons are still relevant.
\end{remark}

\section{Blocking vs. complete randomization\ldots}\label{sec:bk_vs_cr}

With the tools from the previous section, we now turn to comparing complete randomization to blocking.
In the following sections, we systematically explore the difference in variance of our two designs in the finite sample and under a number of superpopulation frameworks.
At the end of this section, Table~\ref{tab:sum} gives an overview about which frameworks provide a guarantee that blocking will be as good or better in terms of precision.

\subsection{\ldots in the finite sample world}\label{sec:bk_vs_cr_fin}

The finite setting is fairly well established.
Here we present a result similar to those previously presented in other papers, such as \cite{imai_king_stuart_2008} and \cite{Miratrix2013}.
In particular, the difference in variance of the treatment effect estimator between the completely randomized design and the blocked design, in the finite sample, is
\begin{align}
&\text{var}\left(\tauestcr|\mathcal{S}\right) - \text{var}\left(\tauestbk|\mathcal{S}\right) \label{eq:cr_vs_bk_fin_samp} \\
&=\frac{1}{n-1}\left[ \text{Var}_k\left(\sqrt{\frac{p}{1-p}}\meanpock + \sqrt{\frac{1-p}{p}}\meanpotk\right) - \sum_{k=1}^K\frac{n_k}{n}\frac{n-\nk}{n}\text{var}(\taukest|\mathcal{S})\right]  \nonumber
\end{align}
\noindent where
\begin{align}
&\text{Var}_k\left(X_k\right) \equiv \sum_{k=1}^K\frac{n_k}{n}\left(X_k - \sum_{j=1}^K\frac{n_j}{n}X_j\right)^2. \label{eq:var_k_def}
 \end{align}

Whether this quantity is positive or negative depends on whether a particular form of between block variation is larger than a form of within block variation.
Finite sample numerical studies in Section~\ref{sec:sims} show an example where even in the worst case for blocking, when all blocks have the same distribution of potential outcomes, the increase in variance is not too great.
A further numerical illustration in Section~\ref{sec:applied_example} further suggests that costs are minimal in practice.

The first term is simply how much variation in means we have across groups.
The second is driven by the average within-group variation.
If the groups are quite variable internally, but are similar in terms of averages, blocking can hurt.

Most prior work state that although the difference in the brackets can be negative, as the sample size grows this difference will go to a non-negative quantity.
However, this statement depends on the type of blocks we have.
In particular, if we have structural blocks such that as $n$ grows, the number of blocks $K$ also grows, the difference in the brackets of Equation~\ref{eq:cr_vs_bk_fin_samp} will not necessarily go to zero or become positive as $n \to \infty$ \citep{Miratrix2013}.
Consider, for example, if $n_k$ is fixed, the above formula's second term is simply a constant times the mean variance.
Therefore a ``bad'' choice of structural blocks, in terms of between block variance being lower than within block variance, could have asymptotic consequences.
This has ties to the random sampling of strata framework, as discussed in Section~\ref{subsec:bk_vs_cr_finstrat}. 

As a final note, we can see the structure of Equation~\ref{eq:cr_vs_bk_fin_samp} more starkly if we consider the case of $K$ equal sized blocks with $n_k = n/K$ units in each, $p = 1/2$, and no treatment effect, with $Y_i(c) = Y_i(t)$ for all units.
In this case we obtain
\begin{align*}
\text{var}\left(\tauestcr|\mathcal{S}\right) - \text{var}\left(\tauestbk|\mathcal{S}\right)
	&=\frac{1}{n-1}\left[ 4\text{Var}_k\left( \meanpock \right) -
	    \frac{4(K-1)}{n} \frac{1}{K}\sum_{k=1}^K  \Sck \right] .
\end{align*}
This simplification makes the comparison of within versus between block variability more clear.

\subsection{\ldots with simple random sampling and flexible blocks}

Given a simple random sample of units, the variance for the completely randomized design yields the following clean and well known result (see \cite{CausalInferenceText}):
\[\text{var}(\tauestcr|\model_{\text{SRS}}) = \frac{\sigmat}{n_{t}} + \frac{\sigmac}{n_{c}},\]
where $\sigmaz$ is the variance of potential outcomes in the (infinite) superpopulation under treatment $z \in \{t,c\}$.

Flexible blocking in this context would be to divide the units into groups using any baseline covariate information we might like, after the sample of units is obtained.
This common setting is what might happen in, for example, experiments that have an initial recruitment drive, with the researchers then tailoring the design of their experiment to the obtained sample.
In these cases we might group by some categorical covariate, aggregating those types of units too few in number to make their own blocks into a single ``left over'' block.
Or we might form blocks out of a continuous covariate by clustering units together as best we can based on the observed sample covariate distribution.
We denote this population framework by conditioning on $\model_{\text{SRS}}$.

The difference of variances between blocking and complete randomization is simply the expectation over Equation~\ref{eq:cr_vs_bk_fin_samp} with respect to the simple random sampling.
In this context, to understand how blocking compares to complete randomization we need to specify the mechanism of how the blocks are formed.
Assume our blocking procedure uses observed baseline covariate information $X$ to group units into blocks and is such that, given a fixed constellation of $X$ values, we will always end up with the same number of blocks with the same values of $X$ within each block.
Further assume that any units with the same values of $X$, indistinguishable to the blocking algorithm, are interchangeable.\footnote{We believe algorithms that will end up with random sizes of blocks, even for the same array of units, should follow this argument with a further conditioning step.  Formal justification of this is beyond the scope of this work.}
Then the expectation of performance over the simple random sampling would also capture the process of making the blocks based on $X$, as a function of the random set of $X$.
The overall difference between blocking and complete randomization will then, in essence, be the difference between how much variation we manage to keep across blocks (as represented by block means) and how much variation is left over within blocks, all averaged across the possible samples.
In other words, Equation~\ref{eq:cr_vs_bk_fin_samp} shows that a good blocking algorithm should reduce heterogeneity  in the $Y_i(z)$ within blocks and maximize variation of $\bar{Y}_k(z)$ across blocks.

To explore the possible harm of flexible blocking, we first consider a scenario where blocking would usually be considered a ``bad idea'': building blocks out of a covariate that is independent from the outcomes.

\begin{theorem}[Simple random sampling and flexible blocking with independent covariates]\label{thm:indep_srs}
If we have a fixed blocking algorithm and the covariates used to form blocks are independent of potential outcomes in the superpopulation, then
$\text{var}(\tauestcr|\model_{\text{SRS}})  = \text{var}(\tauestbk|\model_{\text{SRS}}) $.
\end{theorem}
See~\ref{append:cr_vs_bk_srs} of the Supplementary Materials for proof.
The core idea is that blocking on an $X$ independent of outcome is the same as forming random blocks, which is equivalent to no blocking, i.e. complete randomization. 
In these cases the small amount of random separation of the block means perfectly offsets the blocking cost (these are the two terms in Equation~\ref{eq:cr_vs_bk_fin_samp}).

Interestingly, however, it is possible to do worse when the independence of Theorem~\ref{thm:indep_srs} does not hold: one way to do this is to group units so their group means of $X$ tend to be similar while leaving the within group variances of $X$ high.
If $X$ is highly predictive of $Y$, Equation~\ref{eq:cr_vs_bk_fin_samp} suggests that doing this grouping based on $X$ will result in the same grouping structure on $Y$, producing high within group variability and low between group variability in outcomes.
If we can do this type of grouping well enough for most samples we draw, we would get overall inferior performance under this blocked design.
This is in fact possible, as we illustrate in Section~\ref{sec:flex_block_sim}.
Blocking in this manner would happen due to, for example, the occasional misconception that blocks, rather than the units within the blocks, should be made as similar to each other as possible.

In principle, systematic blocking worsening precision could also happen inadvertently: even if we are systematically reducing variation of $X$ within blocks, if the systematic relationship of $X$ and $Y$ happens to be exactly, perversely, wrong, we can still fail.
We illustrate this as well in the latter part of our simulation Section~\ref{sec:flex_block_sim}, where we designed a relationship where the outcome depended on whether the covariate $X$ was even or odd, and blocked in a manner to ensure balance on this across blocks.
This blocking mechanism substantially reduces within-block variance in our covariate $X$, suggesting gains, but still results in poor performance.
Such inadvertent causing of harm seems unlikely to happen in practice.

These two simulations show that, without further assumptions on the structural relationship between $X$ and the outcomes, we cannot guarantee that blocking is not harmful.
That being said, as our examples suggest, if a blocking procedure tends to group like with like, in terms of $X$, then it is hard to imagine a case that could be worse than blocking on something irrelevant.
That is, it is hard to imagine a case when blocking would cause any harm in this setting.
In our view, the ``Independence'' case used in Theorem~\ref{thm:indep_srs} is a reasonable worst-case for blocking approaches that are not explicitly blocking to minimize cross-block heterogeneity\footnote{\cite{savje2015performance} came to similar conclusions in an investigation of a specific form of flexible blocking called threshold blocking \citep[see][]{higgins2016improving}.
In particular, \cite{savje2015performance} discussed how the blocking may help or harm, depending on the true relationship between the covariates used to block and the outcome and provided a useful decomposition of the true variance in terms of aspects of the blocking algorithm.
However, threshold blocking allows for unequal proportions of units treated within each block, which causes additional complications we consider in Section~\ref{sec:cr_vs_bk_unequal}.
\cite{savje2015performance} showed no harm of another form of flexible blocking which does have equal proportions treated across blocks, fixed-size blocking (which here would be forming matched pairs), in a specific setting with a single binary covariate and blocking in a sensible way to match units with the same covariate values together, to the extent possible.}.

We next turn to simple random sampling for the other types of blocks.
We first note that structural blocks cannot, by design, at all be used in a simple random sampling setting; e.g., we cannot randomly sample $n$ individuals who are twins from an infinite population of twins and expect to find  any complete pairs of twins in our final sample.
Block type is essential in determining which superpopulation framework makes sense for a given experiment, and therefore what guarantees the researcher can realistically expect.

Simple random sampling with fixed blocks is similarly ill defined: what happens if a singleton, or no, units from a given block are sampled due to random chance?
We could extend this case to flexible blocks if we have simple rules for how these partially sampled blocks are combined with the others; if the fixed blocks are few in number, relative to the overall sample size, this adjustment is likely to be only a small number of units, implying that we would obtain substantively similar results.
That being said, we consider the fixed blocks case more carefully in the next section.

\subsection{\ldots with stratified sampling and fixed blocks}\label{subsec:bk_vs_cr_infstrat}

In stratified sampling there is a superpopulation that contains fixed strata, and we sample a fixed number of units from each stratum, and then randomize the units within each stratum independently.
This is the typical type of framework associated with a predefined categorical covariate used for blocking.
For example, a medical trial may recruit a set number of individuals in groups defined by preset age ranges and gender categories.
We denote this population framework by conditioning on $\model_{\text{strat}}$.

Similar to complete randomization under simple random sampling, the variance of the blocking design simplifies to the following result in this setting \citep[see][]{imbens2011experimental}:
\[\text{var}\left(\tauestbk|\model_{\text{strat}}\right)= \sum_{k=1}^{K}\frac{\nk^2}{n^2}\Big(\frac{\sigmatk}{\ntk} + \frac{\sigmack}{\nck}\Big),\]
where $\sigmazk$ is the population variance of potential outcomes under treatment $z \in \{t,c\}$ for units within stratum $k$.

In this context, it is not possible for blocking to be harmful (assuming equal proportion treated in all blocks):

\begin{theorem}[Variance comparison under stratified sampling]\label{prop:cr_vs_bk_strat_samp}
The difference in variance between complete randomization and blocked randomization under the stratified sampling framework is
\begin{align*}
\text{var}&\left(\tauestcr|\model_{\text{strat}}\right) - \text{var}\left(\tauestbk|\model_{\text{strat}}\right)=\frac{1}{n-1}\text{Var}_k\left(\sqrt{\frac{p}{1-p}}\popmeanpock + \sqrt{\frac{1-p}{p}}\popmeanpotk \right)\geq 0 
\end{align*}
where $\popmeanpozk$ is the population mean of potential outcomes under treatment $z$ in stratum $k$ and $\text{Var}_k(\cdot)$ is defined as in Equation~(\ref{eq:var_k_def}).
\end{theorem}

The above expression is very similar to the positive term in Equation~(\ref{eq:cr_vs_bk_fin_samp}) for the finite sample framework. 
Now, however, we no longer have the negative term.

\begin{remark}
Interestingly, when comparing blocking to complete randomization in an infinite population setting, researchers have typically evaluated the completely randomized design under the simpler sampling mechanism of simple random sampling and analyzed the blocked design under the stratified sampling framework \citep[e.g., see][]{imai_king_stuart_2008, imbens2011experimental}.
The found difference between the two estimators is therefore an agglomeration of differences in the characteristics of the samples under the two different sampling regimes as well as the difference in doing a blocked experiment versus a completely randomized experiment.
Specifically, the difference is
\begin{align*}
\text{var}(\tauestcr|\model_{\text{SRS}}) - \text{var}(\tauestbk|\model_{\text{strat}}) \nonumber&= \frac{1}{n_{c}}\sum_{k=1}^{K}\frac{n_k}{n}\left(\popmeanpock-\popmeanpoc\right)^2 + \frac{1}{n_{t}}\sum_{k=1}^{K}\frac{n_k}{n}\left(\popmeanpotk-\popmeanpot\right)^2.
\end{align*}
We see that this difference is also positive, which is where the results claiming that blocking does not cause harm under a superpopulation setting typically comes from.
However, in general, the difference between $\text{var}(\tauestcr|\model_{\text{SRS}}) - \text{var}\left(\tauestcr|\model_{\text{strat}}\right)$ may be positive or negative, showing that the traditional estimates of the benefits of blocking can either be under or overstated in this context.
The result in Theorem~\ref{prop:cr_vs_bk_strat_samp} compares blocking to stratified sampling followed by complete randomization, \emph{not} simple random sampling.
For further discussion, see Supplementary Material~\ref{append:cr_srs_vs_strat_samp}.
\end{remark}

\subsection{\ldots with random sampling of structural blocks}\label{subsec:bk_vs_cr_finstrat}

In this context we sample complete blocks, and then randomize the individuals within the blocks into treatment and control.
Here the blocks are \emph{structural}: they are a consequence of the world, not the researcher's choices.
Examples include twin studies, or a multisite trial with random sampling of sites where each site is then a small randomized trial.
These multisite trials are common in education and medical research where sites may be schools or hospitals.
For instance, with schools we may select schools from an ``infinite'' number of schools, and then randomize treatment to classrooms within each school.
We denote this population framework by conditioning on $\model_{\text{site}}$.

Here the difference between the variances is again the expectation over Equation~\ref{eq:cr_vs_bk_fin_samp} with respect to sampling of blocks:
\begin{align}
&\text{var}\left(\tauestcr|\model_{\text{site}}\right) - \text{var}\left(\tauestbk|\model_{\text{site}}\right) \label{eq:samp_strata} \\
&=E\left[\frac{1}{n-1}\left[ \text{Var}_k\left(\sqrt{\frac{p}{1-p}}\meanpock + \sqrt{\frac{1-p}{p}}\meanpotk\right) - \sum_{k=1}^K\frac{n_k}{n}\frac{n-\nk}{n}\text{var}(\taukest|\mathcal{S})\right]\Bigg|\model_{\text{site}}\right]  \nonumber
\end{align}
As the blocks themselves are sampled with block membership fixed, the expectation can be thought of as over all blocks in the population.
As in the finite framework, because the strata themselves are finite, it is possible that blocking could result in higher variance.
In particular, it is possible to have systematically poor blocks if the block means do not vary.
For example, if we use elementary school classrooms as blocks, we may find that schools break up students into classes such that the classrooms all look similar to each other, in that they have similar proportions of high and low achieving students, but by the same token have higher within classroom variability.

This framework is typically adopted under comparisons of matched pairs (a blocked experiment with blocks of size two) and complete randomization.
Therefore researchers studying matched pairs designs often note that the design can hurt in the superpopulation setting \citep[e.g.,][]{Imai2008}.
However, we make clear here that this harm is not due to the size of the blocks, but rather the block types and sampling scheme.

\subsection{\ldots with two stage sampling}\label{subsec:bk_vs_cr_two_stage}
We next extend our prior setting by letting the sampled strata be themselves infinite in size.
The two stage sampling scheme then works as follows: first randomly select $K$ blocks, then randomly select $n_k$ units within the $k$th selected block.
We generally allow the $n_k$ to depend on the block selected, such that we may not know the total sample size beforehand.
For instance, we might imagine first selecting schools from an ``infinite'' population of schools, and then selecting students from an ``infinite'' population of students within each school.
This sampling scheme is popular in education research and has close ties to multilevel modeling and how multisite experiments are evaluated.
We denote this model with $\model_{\text{2-stage}}$.

If we condition on which blocks were chosen in the first stage, we can draw on results from  stratified sampling (Section~\ref{subsec:bk_vs_cr_infstrat}), giving our variance difference of
\begin{align*}
\text{var}&\left(\tauestcr|\model_{\text{2-stage}}\right) - \text{var}\left(\tauestbk|\model_{\text{2-stage}}\right)\\ 
 &=\EE\left[\frac{1}{n-1}\text{Var}_k\left(\sqrt{\frac{p}{1-p}}\popmeanpock + \sqrt{\frac{1-p}{p}}\popmeanpotk \right)|\model_{\text{2-stage}}\right]\\
&\geq 0 . 
\end{align*}

The sampling of units within strata, as compared to taking the entire structural block, guarantees that blocking will not be harmful with respect to variance, as compared to complete randomization.
Of course, in order to obtain this guarantee, the sampling mechanism and population model need to be valid for the study at hand.
That is, we must have large strata with two stage sampling to rely on this result.
For studies where the entire strata are included in the experiment, the setting of Section~\ref{subsec:bk_vs_cr_finstrat} is the one that should be assumed.

\begin{table}[h!]
    \centering
    \begin{tabular}{|c|c|c|c|c|}
    \hline
    & \multicolumn{2}{|c|}{Equal proportions treated} & \multicolumn{2}{|c|}{Unequal proportions treated}\\
    \hline
        \textbf{Framework} & \begin{tabular}{c}\textbf{Blocking} \\
        \textbf{can help?}\end{tabular} & \begin{tabular}{c}\textbf{Blocking } \\
        \textbf{can hurt?}\end{tabular} & \begin{tabular}{c}\textbf{Blocking} \\
        \textbf{can help?}\end{tabular}  & \begin{tabular}{c}\textbf{Blocking} \\
        \textbf{can hurt?}\end{tabular}\\
        \hline
        \hline
        Finite sample & \color{ForestGreen} Yes & \color{Maroon}Yes & \color{ForestGreen} Yes & \color{Maroon}Yes  \\
        \hline
        \begin{tabular}{c}Simple random sampling,\\ flexible blocks \end{tabular}& \color{ForestGreen} Yes & \color{black}  Yes$^*$ & \color{ForestGreen} Yes & \color{Maroon}Yes \\
        \hline
        \begin{tabular}{c}Stratified sampling,\\  fixed blocks \end{tabular}& \color{ForestGreen} Yes & \color{ForestGreen} No& \color{ForestGreen} Yes & \color{Maroon}Yes \\
        \hline
        \begin{tabular}{c}Random sampling of strata,\\ structural blocks\end{tabular} & \color{ForestGreen} Yes & \color{Maroon} Yes& \color{ForestGreen} Yes & \color{Maroon}Yes  \\
        \hline
        \begin{tabular}{c}Two stage sampling,\\ structural blocks \end{tabular}& \color{ForestGreen} Yes & \color{ForestGreen} No& \color{ForestGreen} Yes & \color{Maroon}Yes \\
        \hline
    \end{tabular}\caption{When blocking is guaranteed not to hurt in terms of precision. Results for equal proportions treated appear in Section~\ref{sec:bk_vs_cr} and results for unequal proportions treated appear in Section~\ref{sec:cr_vs_bk_unequal}. $^*$ No harm for making blocks out of irrelevant covariates. Harm is possible if deliberately implementing a self-destructive blocking choice, as discussed in the text.}\label{tab:sum}
\end{table}

\section{The consequences of variable proportion of blocks treated}
\label{sec:cr_vs_bk_unequal}

Our comparison of complete randomization to blocking in the prior section only applies to the small slice of possible experiments in which the treatment proportion is equal across all blocks.
In practice, however, the proportions treated, $p_k$, may be unequal, and in this case the above results are not guaranteed to hold.
In particular, with blocks of variable size, it can be difficult to have the same proportion treated within each block due to the discrete nature of units.

With varying $p_k$, the units within each block are weighted differently than they would be in a complete randomization when calculating a treatment effect estimate.
That is, in a complete randomization, the treated units are all weighted proportional to $1/p$ but here the treated units in each block get weighted instead by $1/p_k$, meaning units with low probability of treatment will ``count more'' towards the overall treatment mean and their variability will have greater relevance for the overall variance of the estimator.
This can be seen easily from the estimator formulation:
\begin{align*}
\tauestbk &=  \sum_{k=1}^K\frac{\nk}{n}\taukest\\
& =  \sum_{k=1}^K\frac{\nk}{n}\left(\frac{1}{p_k\nk}\sum_{i: b_{i}=k} \mathbb{I}_{Z_i = t} \pot  - \frac{1}{(1-p_k)\nk} \sum_{i: b_{i} =k} (1-\mathbb{I}_{Z_i = t}) \poc\right)\\
& =  \sum_{k=1}^K\left(\frac{1}{n_t}\frac{p}{p_k}\sum_{i: b_{i}=k} \mathbb{I}_{Z_i = t} \pot  - \frac{1}{n_c}\frac{1-p}{1-p_k} \sum_{i: b_{i} =k} (1-\mathbb{I}_{Z_i = t}) \poc\right).
\end{align*}
\cite{savje2015performance} also discussed the effect of variable proportions treated on variance and \cite{higgins2015blocking} explored estimators for blocked designs with possibly unequal treatment proportions, but also multiple treatments.
The costs here are similar to the costs of variable selection probabilities in survey sampling \citep[see][]{sarndal2003model}.

When different blocks have different proportions of units treated, it is possible to systematically have blocks and treatment groups with more variance to also have more weight, which could cause blocking to be harmful even in the stratified sampling setting of Section~\ref{subsec:bk_vs_cr_infstrat} where we usually have guarantees on the benefits of blocking.
In the face of unequal proportions treated, we have the following result for the stratified sampling context:  
\begin{theorem}[Variance comparison with unequal treatment proportions]\label{theorem:unequal_prop}
\begin{align*}
&\text{var}\left(\tauestcr|\model_{\text{strat}}\right)-\text{var}\left(\tauestbk|\model_{\text{strat}}\right) \nonumber\\
  &=\frac{1}{n-1}\text{Var}_k\left(\sqrt{\frac{p}{1-p}}\popmeanpock + \sqrt{\frac{1-p}{p}}\popmeanpotk \right)+\sum_{k=1}^K\frac{(p-p_k)n_k}{n^2}\left[\frac{\sigmack}{(1-p_k)(1-p)}-\frac{\sigmatk}{p_kp}\right] .
\end{align*}
\end{theorem}
The first term in the above result is exactly the usual difference in the stratified sampling setting with equal proportions, as shown in Theorem~\ref{prop:cr_vs_bk_strat_samp}, and is always non-negative.
The second term, however, can be positive or negative depending upon the proportion treated within each block and the variability of potential outcomes under treatment and control within each block.
So we see that with unequal proportions treated across blocks, we do not have a setting with simple guarantees on the benefits of blocking.

Consider the following simplification when we have constant additive treatment effects within each block such that $\sigmack = \sigmatk = \sigma^2_k$:
\begin{align*}
&\text{var}\left(\tauestcr|\model_{\text{strat}}\right)-\text{var}\left(\tauestbk|\model_{\text{strat}}\right) \nonumber\\
  &=\frac{1}{n-1}\text{Var}_k\left(\sqrt{\frac{p}{1-p}}\popmeanpock + \sqrt{\frac{1-p}{p}}\popmeanpotk \right)+\sum_{k=1}^K\frac{n_k}{n^2p(	1-p)}\left[\frac{(p-p_k)(p-(1-p_k))}{(1-p_k)p_k}\right] \sigma^2_k.
\end{align*}
The second term can be negative if, for example, $1 - p_k < p < p_k$ or $p_k < p < 1-p_k$ for all the blocks.
This could occur, for instance, if $p = 0.5$ but in some blocks the researcher is unable to treat exactly half the units (e.g., because of odd number block sizes).
In this circumstance, if our blocks are such that $\popmeanpock = \popmeanpoc$ and $\popmeanpotk = \popmeanpot$ for all $k$, the whole expression will be negative.
Without the simplification that $\sigmack = \sigmatk = \sigma^2_k$, the effect of the varying proportion treated can be mitigated or exacerbated depending upon whether systematically more or fewer units are allocated to the more variable treatment condition within the blocks.

In our simulations in Section~\ref{sec:sims} we in fact see degradation in the benefits of blocking on weakly predictive covariates when the proportion treated is only approximately equal, rather than precisely equal, for the finite sample.
This comes from the additional variation induced by different units being weighted differently, and is akin to the increase in variance from weighted survey sampling.
This suggests that in many realistic scenarios, one might not want to block on covariates that are only weakly predictive.
Section~\ref{sec:sims} provides simulations and discussion to illustrate how blocking can be harmful, or helpful, in this case.

\section{Misconceptions on the comparison of blocking and complete randomization}\label{sec:ignore_bk}

In this section we explore two misconceptions we have encountered regarding comparisons of blocking and complete randomization. 
First, there is a common belief that one can simply ignore blocking that was done and analyze a blocked experiment as a completely randomized one without consequence.
This belief is a misconception in some contexts: implementing a blocked design and then ignoring the blocking when calculating the variance estimator will not necessarily be conservative for the variance of $\tauestbk$.
Secondly, there is a belief that the completely randomized variance estimator (under a completely randomized design) is guaranteed to have lower variability than the typical blocking variance estimator (under a blocked design).
This also does not hold in some contexts.

Before exploring these misconceptions in the next two sections, we first need to introduce the standard variance estimators.
For a completely randomized design, the standard variance estimator is
\begin{align*}
\varestcr  = \widehat{\text{var}}(\tauestcr) =\frac{\scest}{n_{c}} + \frac{\stest}{n_{t}},
\end{align*}
where $\szest$ is the estimated sampling variance for units assigned to treatment $ z \in \{t,c\}$.
Under complete randomization, this estimator is conservative for the finite sample, with bias $\Stc/n$, and is unbiased under simple random sampling \citep[see, e.g.,][]{CausalInferenceText}.

Variance estimation for the blocked design is a bit more complicated.
However, the usual variance estimation strategy for blocked experiments with at least two treated and two control units in each block is to estimate the variance within each block separately, as in the completely randomized design, with
\[\varestk  = \widehat{\text{var}}(\taukest) =\frac{\sckest}{\nck} + \frac{\stkest}{\ntk}.\]
Here the $\szkest$ are analogous to $\szest$ within each block $k$.
These variance estimators can then be combined into an overall variance estimator of
\[\varestbkone  = \widehat{\text{var}}(\tauestbk) = \sum_{k=1}^{K} \frac{\nk^{2}}{n^{2}}\left(\frac{\sckest}{\nck} + \frac{\stkest}{\ntk}\right).\]
Extending results for the completely randomized variance estimator, it is well known \citep[see, e.g., ][]{imbens2011experimental} that $\varestbkone$ is conservative in the finite sample and unbiased under stratified random sampling.
See \cite{pashley_main} for a full discussion of variance estimation under the blocked design, including variance estimators for designs including blocks with a singleton treatment or control unit, and the resulting bias in estimation under different population frameworks.

\subsection{Misconception I: Completely randomized variance estimators are conservative for the blocked design}\label{sec:ignore_block}

Here we consider what happens if we ignore blocking done in the design stage when analyzing the data.
First, if we do not have $p_k=p$ for all $k$, then it is possible that ignoring the blocking structure can cause bias, with $\EE\left[\tauestcr| \mathcal{S}, \textbf{P}_{blk}\right] \neq \taufs$.
$\tauestcr$ could be biased, even under a constant treatment effect assumption, because the completely randomized estimator will effectively give higher or lower weight to some units than the blocked estimator. 
In this case, variance comparison is less relevant.

When $p_k=p$, although bias in the treatment effect estimator is no longer a concern, using the standard completely randomized variance estimator under a blocked design can cause issues.
Firstly, the correct analysis follows from the experimental design.
Therefore, if blocking was implemented it should be taken into account in the analysis.
If standard errors based on the completely randomized design are used instead, we are using standard errors that do not actually reflect uncertainty in the design we ran.
In other words, the standard errors for the completely randomized design are irrelevant for the blocked experiment we actually ran!

Now we turn to a more technical discussion of why researchers may believe the blocks can be ignored and why this belief is flawed.
In particular, the estimated standard errors for a blocked design, when there are many blocks and/or few units, can sometimes be unstable.
This instability is separate from the true performance of the blocked estimator; the instability is in estimating the uncertainty, not the uncertainty itself.
A researcher might think, therefore, to implement blocking to realize its gains, but then perform the analysis as a completely randomized experiment to avoid these concerns.
Unfortunately, this strategy is not guaranteed to be a good choice.

Regarding this first misconception, we have the following:
\begin{theorem}[Completely randomized variance estimator under blocking: Finite sample]\label{theorem:pretend_cr_finite}
In the finite sample setting, analyzing a blocked experiment as if it were completely randomized could give anti-conservative estimators for variance.
\end{theorem}
That is, it is possible to have $\EE\left[\varestcr| \mathcal{S}, \textbf{P}_{blk}\right] \leq \text{var}\left(\tauestbk|\mathcal{S}, \textbf{P}_{blk}\right)$, where $\textbf{P}_{blk}$ is a blocked randomization assignment mechanism.
See Supplementary Material~\ref{append:use_cr_fin_samp} for a derivation that proves this result (assuming $p_k=p$ for all $k$ and with a positive correlation of potential outcomes).
Section~\ref{sim:misconceptions} illustrates this result with a simulation.

However, in the stratified sampling framework, ignoring blocking when a blocked design was run will always result in a conservative estimator for the variance of $\tauestbk$.
\begin{corollary}[Completely randomized variance estimator under blocking: Stratified sampling]\label{cor:pretend_cr_strat}
Analyzing a blocked experiment as if it were completely randomized will not give anti-conservative estimators for variance if we are analyzing for a superpopulation with fixed blocks and stratified random sampling.
\end{corollary}
In the context of stratified random sampling, ignoring the blocking in variance estimation is ``safe.''
See Supplementary Material~\ref{append:use_cr_strat_samp} for more on this result (assuming $p_k=p$ for all $k$).

\subsection{Misconception II: Completely randomized variance estimators have lower variance}\label{sec:var_cr_vs_bk}
Comparisons of blocking to complete randomization often include a discussion on the performance of the variance estimators in terms of their own variance under each design.
There is a misconception that the variance of the blocking variance estimator will always be larger than that of the complete randomization variance estimator.
We next show that misconception is not necessarily true.
Assume that $n_{z,k} \geq 2$ for all blocks $k$ and all treatment assignments $z$.
We focus on the superpopulation framework with stratified random sampling.
The question is, do we have guarantees that one variance estimator will have lower variance?

If $\text{var}(\scest|\model_{\text{strat}}, \textbf{P}_{cr}) \leq \text{var}(\sckest|\model_{\text{strat}}, \textbf{P}_{blk})$ and $\text{var}(\stest|\model_{\text{strat}}, \textbf{P}_{cr}) \leq \text{var}(\stkest|\model_{\text{strat}}, \textbf{P}_{blk})$ for most $k = 1,..., K$, then the completely randomized variance estimator will have lower variability than the blocking variance estimator. 
\cite{imbens2011experimental} gives such an example.\footnote{We note, however, that we disagree with the generalization made in that paper about the variability of the variance estimators, as explained in this discussion.}
In \citeauthor{imbens2011experimental}'s example, the potential outcomes under control have no variation ($\sigmack = 0$ for all $k = 1,...K$) and the distribution of the potential outcomes under treatment is the same in all of the strata ($\sigmatk = \sigmat$ for all $k= 1,..,K$). 
Then, Imbens argues, because $\stest$ is a less noisy estimator of $\sigmat$ than any of the $\stkest$, the variance of the variance estimator would be smaller under the completely randomized design than under the blocked design.
The notion that the variance estimator for complete randomization is less noisy because it is using more data may be true in many situations.

However, this result does not hold in general. 
For instance, consider a population with four strata. 
Within each stratum, there is zero treatment effect and all units are identical. 
Between strata, however, the potential outcomes differ. 
For convenience, say in stratum one $\poc = \pot = 1$, in stratum two $\poc = \pot = 2$, in stratum three $\poc = \pot = 3$, and in stratum four $\poc = \pot = 4$. 
Now assume that four units are sampled from each of the strata. 
In a blocked design, our variance estimate would always be 0. 
But in a completely randomized design, the variance estimate would change based on which units were assigned to treatment and control. 
Thus, the blocking variance estimator would have 0 variance whereas the completely randomized variance estimator would have non-zero variance.

To further explore this question, we compare the variances of the standard variance estimators (under their respective designs) in a simulation in Section~\ref{sim:misconceptions}.
We find that as the blocking estimator gets more precise, relative to complete randomization, the precision of the associated variance estimator also improves.
In the case where blocking is not beneficial, we do see a slight increase in the instability of the variance estimator.
Overall we see the relative uncertainty of the blocking estimator is proportional to the true variance of the blocking estimator, and so when the true variance goes down, the uncertainty of estimating that variance goes down as well.
See Section~\ref{sim:misconceptions}.
In general we do not find this additional instability, when blocking is ineffectual, to be a concern; more serious, perhaps, would be the impact of a degrees of freedom adjustment when doing inference in the case of experiments with a small number of small blocks.


\section{Simulations}\label{sec:sims}

We underscore the findings and arguments in our theoretical work with a few illustrative simulations.
Our first set of simulations illustrate the general value of blocking along with its potential cost, as a function of how successful the researcher is in separating out relatively homogenous sets of units.
The second set of simulations unpacks the difficulties in obtaining overall theoretical results on flexible blocking by showing a range of scenarios, including one that might deceptively look beneficial, where flexible blocking can hurt.
The third set of simulations illustrate the misconceptions discussed in the previous section.

\subsection{Cost and benefit of blocking as a function of block variation}\label{subsec:bk_vs_cr_sim}

\begin{figure}
\centering
\includegraphics[width=0.75\textwidth]{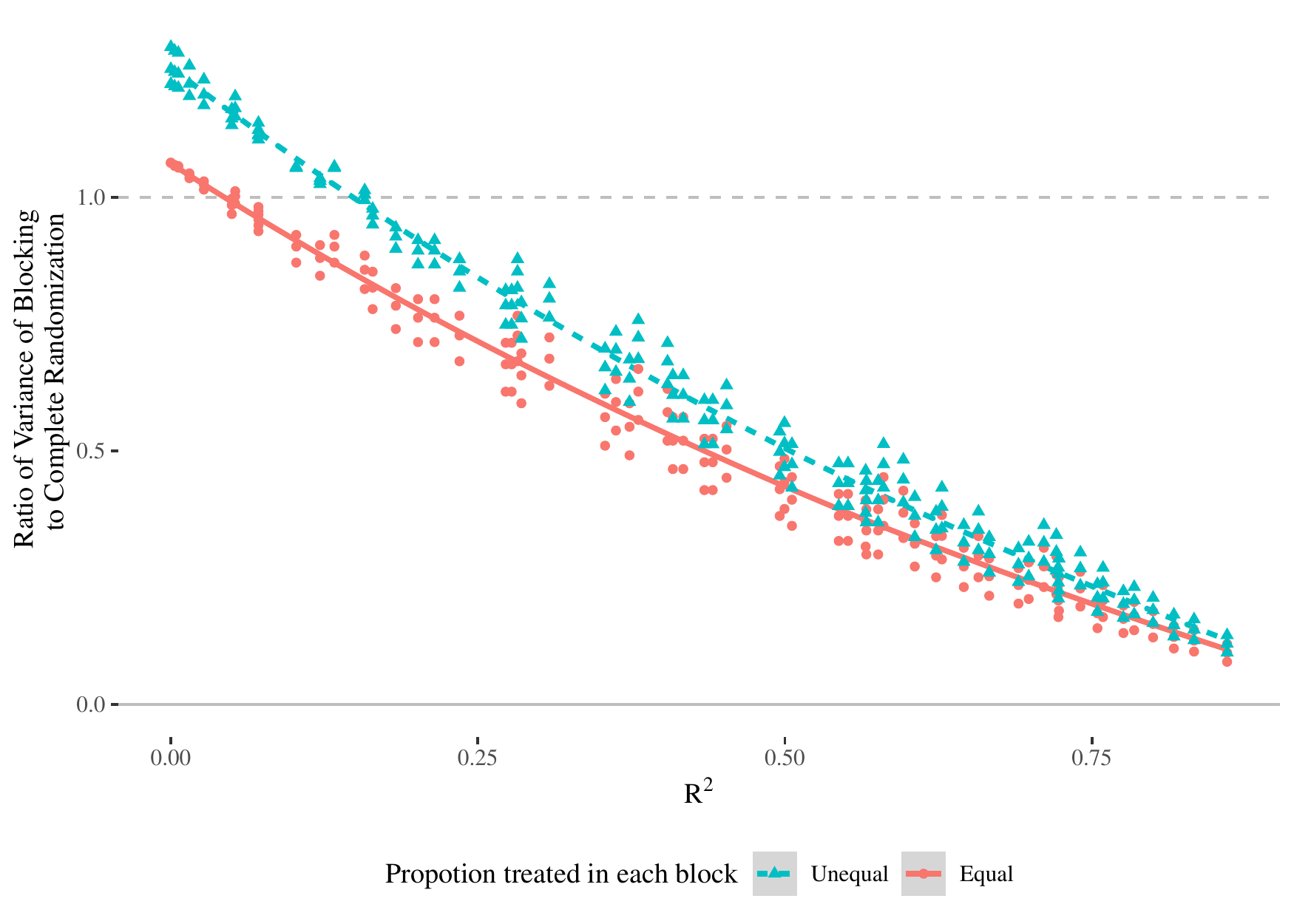}
\caption{Numerical study to assess completely randomized versus blocked design in finite sample context, when $p_k = 0.2$ (equal proportions) or unequal proportions across blocks.
The y axis is $\text{Var}(\tauestbk|\mathcal{S})/\text{Var}(\tauestcr|\mathcal{S})$.}
\label{plot:bk_vs_cr}
\end{figure}
\color{black}

Following our theoretical results, we compare actual variances of the treatment effect estimators to avoid complications of any costs or differences in estimating these variances.
We examine a series of scenarios ranging from a collection of blocks where there is little variation from block to block (causing blocking to be less beneficial) to scenarios where the blocks are well separated and blocking is critical for controlling variation.
In our first numerical study we treat 20\% of units in all of the blocks, with specific block sizes of 10, 15, and 20.
We had a total of 8 blocks.
In the second numerical study we keep the sizes and number of blocks the same, but allow the proportion treated to vary from block to block, from 0.1 to 0.3, with no block having exactly 20\% treated. 
We keep the overall proportion treated the same so that the completely randomized design was the same in both numerical studies.
Because the number of treated units is forced to be a positive integer, the small blocks of size 10 have more disparate proportions (.1 or .3), the blocks of size 15 have slightly less disparate proportions (2/15 or 4/15), and the blocks of size 20 have the least disparate proportions (0.15 or 0.25).
Different set ups for these unequal proportions may yield different results.

Our data generating mechanism generally follows the one presented in \cite{pashley_main}.
The simulations are for the finite sample.
Data was generated once from normal distributions for each setting but guaranteeing that the empirical averages and variances for each block match the theoretical values set for that simulation.
This avoids possible pitfalls of odd behavior from a single random finite sample.
The block means and treatment effects were set such that they had a negative correlation with block size; larger blocks had lower potential outcomes and smaller treatment effects than small blocks.
Through the simulation we varied how close the blocks were in terms of control means and treatment effects.
We also varied the correlation of potential outcomes within blocks as $\rho = 0, 0.5, 1$.
Correlation of 1 corresponds to additive treatment effects within each block.
The variances (and variance ratios) were calculated based on the full schedule of potential outcomes, using the formulas presented in Section~\ref{subsec:est_var}.

The first numerical study corresponds to the mathematical argument presented in Section~\ref{sec:bk_vs_cr_fin} and examines how much blocking can hurt.
The second numerical study illustrates what changes when proportions treated are not the same across blocks, as discussed in Section~\ref{sec:cr_vs_bk_unequal}.

We see on the $x$-axis of Figure~\ref{plot:bk_vs_cr} an $R^2$-like measure of how predictive blocks are of outcome, calculated for each finite data set investigated.
The $R^2$ measure was varied by manipulating the spread of block means under control and the spread of block treatment effects.
That is, the x-axis tells us how ``good'' our blocking is.
If we were blocking based on a covariate value, this would measure how predictive that covariate is of the outcome (higher $R^2$ means more predictive).
The $y$-axis is the ratio of variances of blocking versus complete randomization; values above 1 (dashed horizontal line) indicate a cost to blocking and values below 1 indicate a benefit to blocking.

Generally, as expected, in most scenarios we find blocking to be helpful.
We see large gains in blocking for moderate to large $R^2$, and only a slight penalty to blocking when the $R^2$ is relatively small.
That is, even in the extreme case where blocks are formed in a (unrealistically) poor manner such that all blocks have exactly the same average potential outcomes but there is variability within in each block, we observed at most an increase in variance of about 7\% for using blocking with equal proportions.
On the other hand, the benefits of blocking go up to an almost 92\% reduction in variance in the most extreme scenario considered.

When we vary the proportions treated, the gains of blocking are muted.
The maximum variance increase using blocking in this case is 30\%, which occurs in the ``worst case'' blocking scenario.
This additional cost of blocking is due to the inability to weight all units equally because of the variable proportions treated within each block.
This is analogous to the additional cost of incorporating weights in, for example, survey experiments \citep[][]{miratrix_survey}.
The maximum benefit of blocking was similar to the equal proportion case, with an approximately 89.8\% reduction in variance using blocking.

These results provide some practical guidance: If our blocking factor is even moderately predictive of outcomes, we expect blocking to be beneficial.
With equal proportions treated across all blocks, even if we have poor blocks, the cost of blocking will not be too large.
When we move away from equal proportions being treated, the story gets a bit murkier in terms of exactly how harmful blocking can be, but we still expect gains from blocking as long as our blocking factor is reasonably predictive of outcomes.

\subsection{Costs and perils of flexible blocking}
\label{sec:flex_block_sim}

\begin{table}[ht]
\centering
\begin{tabular}{l |rrr|cccc|}
  \hline
				 & \multicolumn{3}{c|}{Relative variance}		 & \multicolumn{4}{c|}{Variance ratios}  \\ 
blocking method & Indep & Linear & Odd & $X$ & $Y$ (Indep) & $Y$ (Linear) & $Y$ (Odd) \\ 
  \hline
Flex & 100.5 & 8.2 & 98.9 & 0.5 & 76.4 & 0.5 & 75.4 \\ 
Interleave & 99.3 & 110.4 & 99.7 & 93.8 & 76.1 & 93.8 & 76.2 \\ 
Peevish & 100.9 & 31.4 & 110.4 & 7.6 & 76.2 & 7.6 & 93.6 \\ 
   \hline
\end{tabular}
\caption{For first 3 columns, numbers are relative percent of the average size of the standard error of the blocking approach versus complete randomization. Each column corresponds to a data generating process dictating the relationship of $X$ to Y.  Each row corresponds to a form of blocking. The final 4 columns show the ratio of average variance of $X$ within block to overall, and the ratio of average variance of $Y$ within block to overall for each DGP considered.}
\label{tab:flex_block_sim_results}
\end{table}

To illustrate the benefits and perils of flexible blocking we designed three data generating scenarios, ``Linear,''  ``Indep,'' and ``Odd,''  dictating the relationship between $X$ and $Y$. In all three $X$ ranges as an integer from 1 to 16.  For ``Linear,'' $Y$ is linearly related to $X$, making $X$ a good thing to block on. For ``Indep'' $Y$ is independent of $X$, making $X$ a useless thing to block on.
Finally, for ``Odd,'' $Y$ is high if $X$ is odd and low if $X$ is even; blocking on $X$ in this case has unclear benefit.

We also examined three methods for blocking.
Our most principled, ``Flex,'' divides the units up by $X$, after sorting $X$, in a manner that ensures each block has an even number of units. 
``Interleave'' makes blocks by interleaving units, trying to make a collection of blocks that are as similar to each other as possible; this is doing blocking the exact wrong way.
Finally, ``Peevish,'' is our perverse, existence-proof blocking method where we group units by similar values of $X$, but also ensure there are the same number of odd and even units in each block.

For each of these 9 combinations we repeatedly (10,000 times) generated data and analyzed it via a simulated blocked experiment and a simulated completely randomized experiment.
Thus, this corresponds to the simple random sampling setting.
We impose a strict zero treatment effect for all units in the experiment, meaning we can measure within versus between block variance on outcomes without worrying about the treatment assignment (as the variance is equal in the two treatment groups).
We then looked at the relative standard error of the blocking approach compared to complete randomization.
Numbers above 100\% means blocking performs worse than doing nothing. Numbers below 100\% show blocking to be beneficial.

Results are given in Table~\ref{tab:flex_block_sim_results}.
When $X$ is independent, how we block does not matter. 
Blocking neither harms nor helps.
Interleaving (deliberately making heterogeneous blocks with respect to $X$) can cause harm if $X$ matters: when $X$ is linearly related to $Y$ we have a variance increase.
And finally, as a proof of concept we show that for our ``Peevish' blocking approach, even though we are reducing the within-block variance of $X$ we are doing it in a manner that exactly fails given the odd data generating process we consider: while this seems unlikely to happen in practice, this demonstrates that, in principle, one could get hurt by a blocking algorithm even while it looks like things are improved.
This combined with our ``Odd'' DGP shows that we can not get any mathematical guarantees on blocking causing no harm without further assumptions on the DGP or restrictions on the blocking approach.

\subsection{Illustration of the misconceptions of Section~\ref{sec:ignore_bk}}
\label{sim:misconceptions}

We next extend the analysis of our simulations to briefly illustrate our two misconceptions in Section~\ref{sec:ignore_bk}.
The simulation setting is the same as Section~\ref{subsec:bk_vs_cr_sim}, but we only look at the case with equal proportions treated across blocks.
Because both of the misconceptions dealt with variance estimation, for each finite population we simulated 5000 randomizations and calculated the various variance estimators for each randomization.
For the first misconception considered in Section~\ref{sec:ignore_bk}, we explore the results of using the completely randomized variance estimator as an estimator for the blocked design.
For the second misconception considered in Section~\ref{sec:var_cr_vs_bk}, we compare the variance of the standard Neyman variance estimator used with the completely randomized design to the standard extension of the Neyman variance estimator to blocking with the blocked design.
Note that we have at least two units in each block with equal proportions treated, allowing the usual blocking variance estimator to be used here.

The results are shown in Figure~\ref{plot:misconceptions}.
The graphs are restricted to low values of $R^2$ to showcase where the more interesting results occur.
In the left graph we see that, for some scenarios considered with a very low $R^2$ value, using the completely randomized variance estimator with a blocked experiment can result in underestimation (in expectation) of the true variance of the blocked experiment; this is illustrated by points falling below the line at 1. 
This demonstrates that it is possible for this approach to be anti-conservative (although only slightly so).
More seriously, this estimator rapidly becomes substantially \emph{conservative}; we do not advise analyzing a blocked experiment as if it were completely randomized.
For comparison, the standard blocking variance estimator is also shown.
The blocking variance estimator is conservative if there are additive effects within each block but the completely randomized variance estimator quickly becomes more conservative at an $R^2$ less than 0.1.

On the right hand side of Figure~\ref{plot:misconceptions}, we compare the variability of the standard complete randomization and blocking variance estimators for a series of experiments. We see that while it is possible, when blocks are relatively homogeneous, for the complete randomization estimator to have less instability, it is generally substantially more unstable.
The relative instability of the completely randomized variance estimator increases as $R^2$ increases and thus the blocking has lower true variance.

 \begin{figure}
 \centering
 \begin{subfigure}{.5\textwidth}
   \centering
   \includegraphics[width=.9\textwidth]{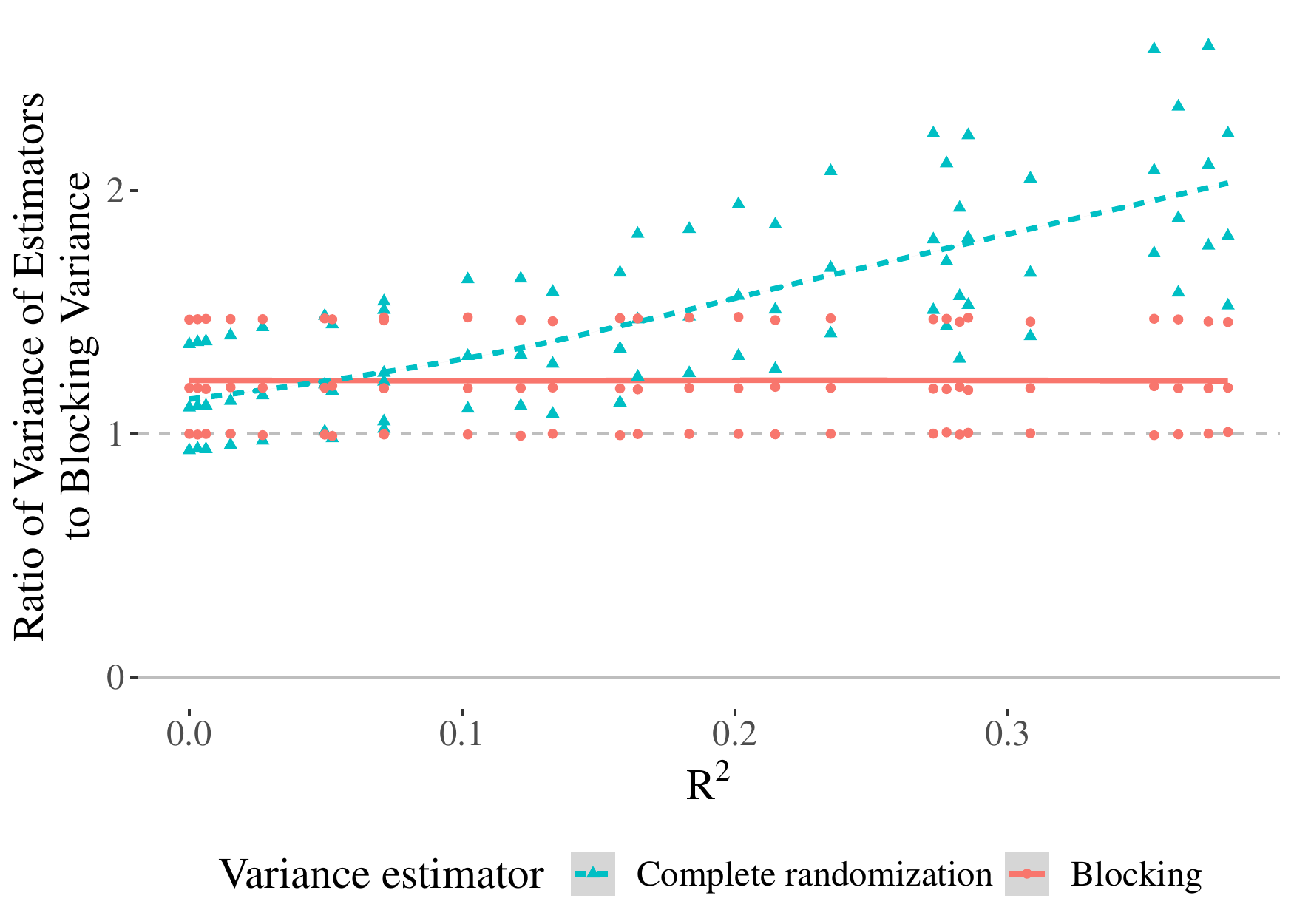}
   \caption{Ratio of estimated variance estimators to the actual blocking variance.}
   \label{fig:sub1}
 \end{subfigure}%
 \begin{subfigure}{.5\textwidth}
   \centering
   \includegraphics[width=.9\textwidth]{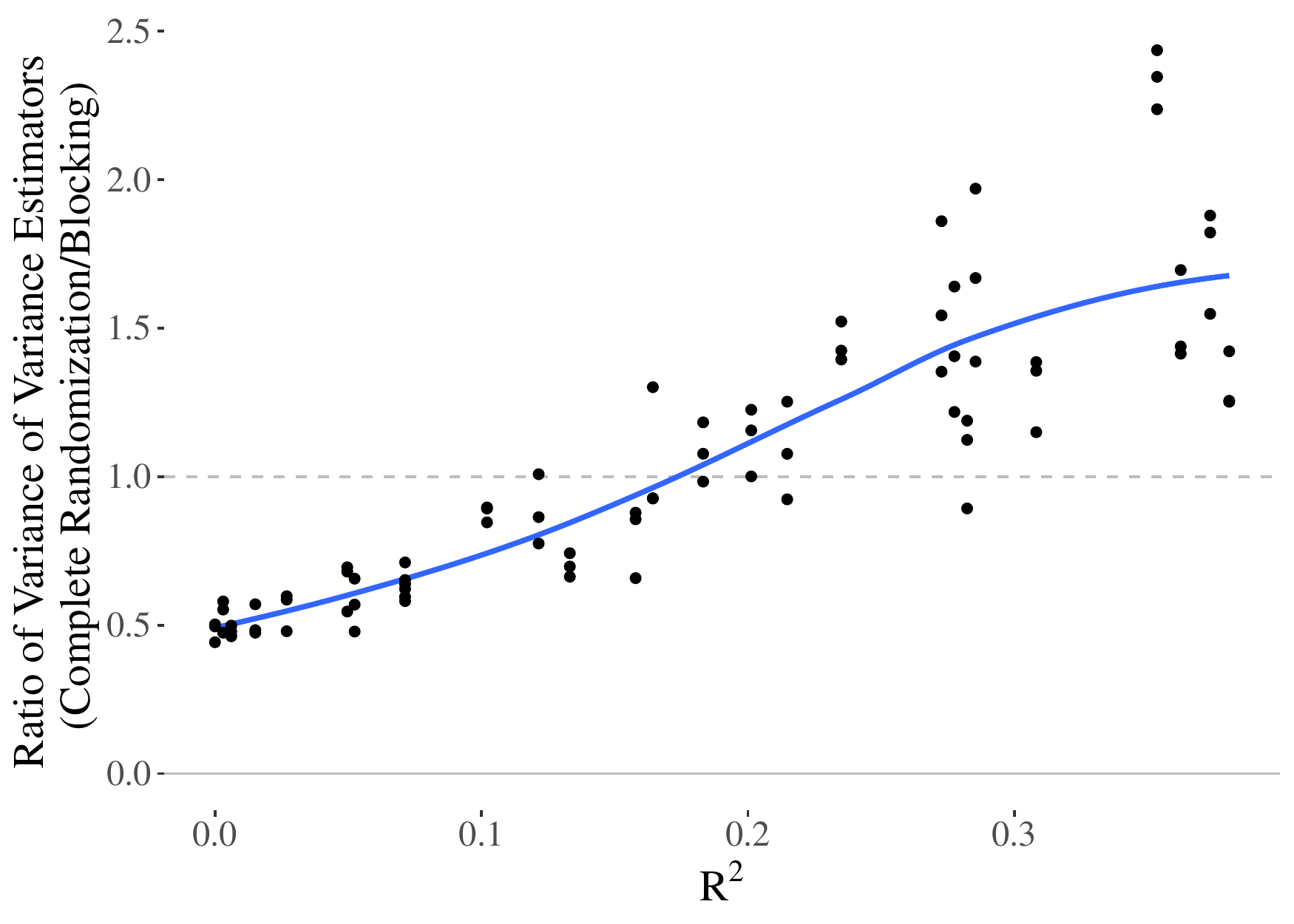}
   \caption{Variance of variance estimators.}
   \label{fig:sub2}
 \end{subfigure}

\caption{Relationship of estimated variance to actual and precision in the estimated variance. At left we see the variance estimator for complete randomization can be lower, on average, than the true blocking variance in some circumstances. At right we see the complete randomization variance estimator can be slightly less, or substantially more variable than the blocking variance estimator.} 
 \label{plot:misconceptions}

 \end{figure}

\section{Teacher professional development, an applied example}
\label{sec:applied_example}

To further explore the benefits and costs of blocking in a more natural context, we used a data set  from a previous randomized trial of a teacher professional development program \citep{Heller:2012ih} to explore a few hypothetical scenarios of how blocking might work or fail to work in practice.
In the actual experiment, the blocks of teachers were defined by geographic area, and were thus outside the control of the researcher.
We have 15 cohorts (blocks) of teachers, with 52 to 87\% of units treated.
Cohort sizes ranged from 8 to 29 teachers.
Here, cohort membership explained only about 10\% of the variance in the baseline test.
Furthermore, due to different pragmatic considerations, the proportion treated in each block varied considerably.\footnote{The original experiment was also a trial of multiple versions of treatment that varied by site, which we have collapsed, creating further imbalance in the block sizes and proportions treated. We also imputed some missing values.}
This is a good example of the worst sort of blocking one might have in a natural context (the blocking was externally forced, hence the non-ideal design).

This experiment provides an opportunity in that three tests---a pre-test and two follow-up tests---were given to the teachers; we use these outcomes to investigate how different hypothetical blocking strategies may have worked differently.
In these investigations, we evaluate how well blocking would work for different hypothetical experiments on units with the same covariates and outcomes as the real units.
We initially use the second test as the baseline and the third test as a final outcome, and assume no actual treatment impact, using the third test as both the potential outcome for control as well as treatment.
This gives us a fully ``observed'' set of potential outcomes on which we can calculate the impacts of different hypothetical design decisions.
That is, we ignore the treatment assignment of units that actually occurred, merely using these data to provide realistic values for an outcome-covariate relationship we might see in practice.
Thus this a numerical illustration based on practical data values, rather than a valid re-analysis of the original study.
We consider several different (hypothetical) blocking decisions, constraining ourselves by the block sizes and treatment proportions actually used.
We then use the formulae from this work to calculate what the standard errors would be under the different blocking strategies considered, holding the schedule of potential outcomes (the third test) and our assumed baseline test (the second test) fixed.

The different explored blocking strategies, and the relative costs and benefits of those strategies, are on Table~\ref{tab:emperical_example}.
The second and third sets of columns show the same exercise, but using different tests as baseline and outcome.
We report on the 2nd and 3rd test as in this case there is no actual treatment delivered between tests, arguably making the relationship between baseline and outcome more natural.
That being said, the trends are nearly identical for alternate configurations.

We first compare the existing design (blocking on geography) versus a hypothetical complete randomization across all units.
Under the finite-sample context, we find a modest 4\% increase in the standard error due to having blocked in this haphazard fashion.

In the positive direction, an experiment where we use the same set of block sizes and proportions treated as our original experiment, but group the teachers into these blocks by their similarity on their baseline test, we find massive benefits to blocking, with the blocked estimator having standard errors 67\% of the size of complete randomization; this type of blocking (on a baseline test) would be a natural choice for experimenters with control over their design.
For reference, the $R^2$ measure of our baseline test on the outcome was about 60\%.  
Near the limit, if we were somehow able to obtain and block on a baseline test that was perfectly correlated with the outcome, our standard errors fall to a quarter of the size of complete randomization.

In the other direction, blocking randomly, using the given pattern of blocks and proportions treated, gives around a 5\% increase in the standard errors compared to complete randomization, with some variation depending on how lucky or unlucky we are in grouping similar units by chance.
The 99th percentile worst allocation, out of 1,000 random allocations tried, was an 8\% increase.
Identifying the worst of a collection of blocking schemes is similar to the minimax analysis of \cite{nordin2020properties}, where they investigate how various restrictions on randomization can lead to risk of higher variance experiments when one has no covariates predictive of outcome.

Differential rates of treatment can have a cost.
To explore this we examine the case of using the blocks as given, randomizing as equal a proportion of units to treatment as possible across blocks (here 71\% with some variation due to needing to treat integer numbers of units).
This case gives a slight benefit to the original blocking structure, with SEs around a half percentage points smaller.
Even in this non-ideal context, the cohorts were still different enough from each other to offset the potential penalty of blocking.
Randomly blocking with equal assignment is still worse than no blocking, but not much worse; even the 99th percentile of 1,000 trials was only a 2.5\% increase.

Overall, for the hypothetical experiments motivated by these data, the worst blocking choices were not so bad, and the best were quite good. 
This reinforces our overall findings: while it is not true that blocking is always beneficial (here, for a finite sample context), it appears difficult to make it substantially harmful.
The main concern appears to be when the harm of uninformative blocking is amplified by substantially varying proportions assigned to treatment within blocks.
With (roughly) equal proportions of units treated across blocks, the harm of blocking was minimal despite blocks that had little relationship to the outcome.
We note that these explorations are all scenarios with no impact for any unit.
More generally, unequal proportions alone would not necessarily be problematic; as discussed in Section 4, with treatment effect heterogeneity unequal proportions treated could actually help if we happen to assign more units to more variable treatment arms.

\begin{table}[ht]
\centering
\begin{tabular}{lrlrlrl}
  \hline
 Scenario & \multicolumn{6}{c}{Relative Standard Error}\\ 
 Scenario & \multicolumn{2}{c}{T3 (T2)} & \multicolumn{2}{c}{T2 (T1)} & \multicolumn{2}{c}{T1 (T3)}\\ 
  \hline
  Actual design & 104.0\% & 									& 102.1 &	& 99.5 &	\\ 
  Balanced assignment proportions & 99.6 			&		& 99.6  && 98.0  & \\ 
  Randomly blocking units & 104.8 & (108.4) 				& 104.8 & (108.0)& 104.9 & (109.1) \\ 
  Random with balanced assignment & 100.9 &(102.5)			& 100.9	& (102.6)& 100.8	& (102.6)  \\ 
  Blocked by baseline test & 67.1 				&			& 93.7 	& & 88.9 	&  \\ 
  Blocked by max inform baseline test & 25.4	&			& 18.5  & & 27.7  & \\ 
   \hline
\end{tabular}
\caption{Relative size of standard error for blocked vs. complete randomization under a variety of blocking strategies applied to the \citet{Heller:2012ih} example data.  Numbers in parenthesis are the upper 99\% of 1000 random assignments. First pair columns are using test 2 for baseline values and test 3 for hypothetical outcome, second pair are using test 1 for baseline and test 2 for outcome, third pair are test 3 as baseline, test 1 as outcome.}
\label{tab:emperical_example}.
\end{table}

\section{Conclusion}

Different types of blocks and sampling frameworks can change the answer to the question ``Is blocking always beneficial in terms of the precision of my estimators?''
We argue that these varying factors are why the current literature can seem confusing and contradictory.
Overall, our answer is that blocking is often beneficial, but there are many nuances.

We carefully compared complete randomization to blocking, identifying that prior literature has often collapsed the sampling step and randomization step.
Overall, we show that blocking often, but not always, improves precision, and that guarantees about blocking depend on the framework adopted.
Blocking will not reduce precision, compared to a complete randomization, when working in the stratified sampling framework with equal proportion of units treated across blocks, no matter how small the blocks are or how poorly they separate the units.
Similarly, we find that in the simple random sampling setting, given a fixed algorithm for creating flexible blocks, if one makes blocks out of covariates independent of the potential outcomes (the ``bad idea'' scenario), the blocking estimator will also be no worse than complete randomization.
In the other two main frameworks considered, however, blocking is not guaranteed to reduce variance. 
We also show that the variance estimators for blocked experiments do not, as is sometimes believed, generally have higher instability than with complete randomization.


These results assume that the blocks have equal proportions of units treated; if the proportions treated differ, we lose all guarantees that blocking will reduce variance, regardless of framework.
That being said, the simulations and numerical example show the potential for large gains of blocking, even with unequal proportions.
While the cost of blocking on a weakly predictive covariate in this context can be larger, we still did not see substantial losses of precision.

Even when blocking is unlikely to be helpful in terms of precision, there are several reasons an experimenter might block.
First and foremost, an experimenter may simply be forced to block given the context or constraints of the experiment; our work suggests that little sleep should be lost when this occurs.
As the numerical example shows, the negative impact is likely to be small.
The simulations further show that blocking can lead to very large gains and only a small potential cost in finite sample settings with equal proportion treated.
Blocking can also guarantee that one has good balance on those covariates used to make ones blocks; this can increase the credibility of an experiment, regardless of any documented relationship between the covariates used and outcomes.
In general, experiments with observed systematic differences in the treatment and control group are viewed with greater skepticism.
Blocking is good insurance against this concern.

Overall, we advocate for the advice ``thoughtfully block when you can'' to emphasize that blocking is usually beneficial, but must be applied with some thought to avoid edge cases such as inadvertently creating blocks that are equal in distribution.
Where possible, researchers should form blocks out of covariates predictive of outcome.
They might consider blocking on multiple such covariates to increase the likelihood of obtaining a beneficial blocking.
Unless one can predict how the variation in the treated and control units will differ for different blocks, we advise keeping the proportion of units treated similar across blocks.
In terms of analysis, we show that one should analyze as a blocked experiment if blocking was done:  completely randomized variance estimators are not necessarily conservative for the blocked design.

In future investigations, it would be useful to assess other practical concerns with blocked designs.
Two such concerns are (1) degrees of freedom concerns due to the larger number of parameters that need estimation, and (2) further assessment of the stability of variance estimators, which we touched on in Section~\ref{sec:var_cr_vs_bk}.
The tradeoff between decreased precision and reduced degrees of freedom by using blocked or matched pairs designs has been noted by others (e.g., \citeauthor{box2005statistics}, \citeyear{box2005statistics}, p. 93; \citeauthor{imbens2011experimental}, \citeyear{imbens2011experimental}; \citeauthor{snedecor1989statistical}, \citeyear{snedecor1989statistical}, p. 101) and is an important practical limitation to consider when using these designs.
Future work should investigate how the real costs of degrees of freedom loss and instability in variance estimation depend on the experimental design within these frameworks.

\bibliographystyle{apalike}
\bibliography{poststratref}{}

\begin{appendices}
\renewcommand\appendixname{Supplementary Material}

These Supplementary Material contain the derivations of the comparison results, and also the proofs for the theorems regarding ignoring blocking.

\section{Derivations of blocking versus complete randomization differences}\label{append:cr_vs_bk}

Unless otherwise noted, the following section assumes that $p_k = p$ for all $k=1,\hdots,K$.

\subsection{Finite sample, Equation~(\ref{eq:cr_vs_bk_fin_samp})}\label{append:cr_vs_bk_fs}
We start by deriving the result of Equation~(\ref{eq:cr_vs_bk_fin_samp}),
\begin{align*}
\text{var}\left(\tauestcr|\mathcal{S}\right) &- \text{var}\left(\tauestbk|\mathcal{S}\right) \\
=&\sum_{k=1}^K\frac{1}{n(n-1)}\Big[\nk\left(\sqrt{\frac{p}{1-p}}\meanpock +\sqrt{\frac{1-p}{p}}\meanpotk- \left(\sqrt{\frac{p}{1-p}}\meanpoc+\sqrt{\frac{1-p}{p}}\meanpot\right)\right)^2 \\
& - \frac{n-\nk}{n(\nk-1)}\sum_{i:b_i=k}\left(\sqrt{\frac{p}{1-p}}\poc + \sqrt{\frac{1-p}{p}}\pot-\left(\sqrt{\frac{p}{1-p}}\meanpock+\sqrt{\frac{1-p}{p}}\meanpotk\right)\right)^2\Big].
\end{align*}

We have from usual results \citep[][p. 89]{CausalInferenceText}, in addition to independence of treatment assignment within blocks and the assumption that $p_k = p$ for $k=1,...,K$, that
\begin{align*}
\text{var}\left(\tauestcr|\mathcal{S}\right) &= \frac{\Sc}{n_c} +  \frac{\St}{n_t} - \frac{\Stc}{n}= \frac{\Sc}{(1-p)n} +  \frac{\St}{pn} - \frac{\Stc}{n}\\
\text{and}&\\
\text{var}\left(\tauestbk|\mathcal{S}\right) 
&= \sum_{k=1}^K\frac{\nk^2}{n^2}\left(\frac{\Sck}{\nck} +  \frac{\Stk}{\ntk} - \frac{\Stck}{\nk}\right)= \sum_{k=1}^K\frac{\nk^2}{n^2}\left(\frac{\Sck}{(1-p)\nk} +  \frac{\Stk}{p\nk} - \frac{\Stck}{\nk}\right).
\end{align*}

To take the difference, we need to write the complete randomization variance in terms of the block components.
First we look at the expansion of $\Sz$:
\begin{align*}
\Sz &=\frac{1}{n-1}\sum_{i=1}^n\left(\poz - \meanpoz\right)^2\\
&=\frac{1}{n-1}\sum_{k=1}^K\sum_{i: b_i = k}\left(\poz - \meanpozk + \meanpozk-\meanpoz\right)^2\\
&= \sum_{k=1}^K\frac{\nk-1}{n-1}\Szk + \sum_{k=1}^K\frac{\nk}{n-1}\left(\meanpozk - \meanpoz\right)^2.
\end{align*}

Now do the same for $\Stc$:
\begin{align*}
\Stc 
&= \sum_{k=1}^K\frac{\nk-1}{n-1}\Stck + \sum_{k=1}^K\frac{\nk}{n-1}\left(\taukfs-\taufs\right)^2.
\end{align*}

\begin{align*}
&\text{var}\left(\tauestcr|\mathcal{S}\right) - \text{var}\left(\tauestbk|\mathcal{S}\right) \\
&=\frac{\Sc}{(1-p)n} +  \frac{\St}{pn} - \frac{\Stc}{n} - \left[\sum_{k=1}^K\frac{\nk^2}{n^2}\left(\frac{\Sck}{(1-p)\nk} +  \frac{\Stk}{p\nk} - \frac{\Stck}{\nk}\right)\right]\\
&=\frac{\sum_{k=1}^K\frac{\nk-1}{n-1}\Sck + \sum_{k=1}^K\frac{\nk}{n-1}\left(\meanpock - \meanpoc\right)^2}{(1-p)n} +  \frac{\sum_{k=1}^K\frac{\nk-1}{n-1}\Stk + \sum_{k=1}^K\frac{\nk}{n-1}\left(\meanpotk - \meanpot\right)^2}{pn}\\
& - \frac{\sum_{k=1}^K\frac{\nk-1}{n-1}\Stck + \sum_{k=1}^K\frac{\nk}{n-1}\left(\taukfs-\taufs\right)^2}{n}\\
& - \left[\sum_{k=1}^K\frac{\nk^2}{n^2}\left(\frac{\Sck}{(1-p)\nk} +  \frac{\Stk}{p\nk} - \frac{\Stck}{\nk}\right)\right]\\
&=\underbrace{\frac{\sum_{k=1}^K\nk\left(\meanpock - \meanpoc\right)^2}{(1-p)n(n-1)} +  \frac{\sum_{k=1}^K\nk\left(\meanpotk - \meanpot\right)^2}{pn(n-1)} - \frac{\sum_{k=1}^K\nk\left(\taukfs-\taufs\right)^2}{n(n-1)}}_{\textbf{A}}\\
& - \underbrace{\sum_{k=1}^K\Big[\left(\frac{\nk}{(1-p)n^2}-\frac{\nk-1}{(1-p)n(n-1)}\right)\Sck +  \left(\frac{\nk}{pn^2}-\frac{\nk-1}{pn(n-1)}\right)\Stk- \left(\frac{\nk}{n^2}-\frac{\nk-1}{n(n-1)}\right)\Stck\Big]}_{\textbf{B}}
\end{align*}

We have now split our calculation into the between block variation and within block variation pieces.

\textbf{A} is the between block variation:

We expand $\sum_{k=1}^K\nk\left(\taukfs-\taufs\right)^2$ as
\begin{align*}
\sum_{k=1}^K\nk\left(\taukfs-\taufs\right)^2&=\sum_{k=1}^K\nk\left(\meanpotk - \meanpock -(\meanpot - \meanpoc)\right)^2\\
&=\sum_{k=1}^K\nk\left[\left(\meanpotk - \meanpot\right)^2 +\left(\meanpock - \meanpoc\right)^2 - 2\left(\meanpotk - \meanpot\right)\left(\meanpoc - \meanpoc\right)\right].
\end{align*}

So then
\begin{align*}
&\textbf{A}=\frac{\sum_{k=1}^K\nk\left(\meanpock - \meanpoc\right)^2}{(1-p)n(n-1)} +  \frac{\sum_{k=1}^K\nk\left(\meanpotk - \meanpot\right)^2}{pn(n-1)}\\
& - \frac{\sum_{k=1}^K\nk\left[\left(\meanpotk - \meanpot\right)^2 +\left(\meanpoc - \meanpoc\right)^2- 2\left(\meanpotk - \meanpot\right)\left(\meanpoc - \meanpoc\right)\right]}{n(n-1)} \\
&=\frac{\sum_{k=1}^Kp\nk\left(\meanpock - \meanpoc\right)^2}{(1-p)n(n-1)} +  \frac{\sum_{k=1}^K(1-p)\nk\left(\meanpotk - \meanpot\right)^2}{pn(n-1)} + 2\frac{\left(\meanpotk - \meanpot\right)\left(\meanpoc - \meanpoc\right)}{n(n-1)}\\
&=\frac{1}{n-1}\sum_{k=1}^K\frac{\nk}{n}\left(\sqrt{\frac{p}{1-p}}\meanpock +\sqrt{\frac{1-p}{p}}\meanpotk- \left(\sqrt{\frac{p}{1-p}}\meanpoc+\sqrt{\frac{1-p}{p}}\meanpot\right)\right)^2\\
&=\frac{1}{(n-1)}\text{Var}_k\left(\sqrt{\frac{p}{1-p}}\meanpock + \sqrt{\frac{1-p}{p}}\meanpotk\right).
\end{align*}

\textbf{B} is the within block variation:

\begin{align*}
\textbf{B}&=  \sum_{k=1}^K\Big[\left(\frac{n-\nk}{(1-p)n^2(n-1)}\right)\Sck +  \left(\frac{n-\nk}{pn^2(n-1)}\right)\Stk - \left(\frac{n-\nk}{n^2(n-1)}\right)\Stck\Big]\\
&= \frac{1}{n^2(n-1)}\sum_{k=1}^K(n-\nk)\nk\Big[\frac{\Sck}{(1-p)\nk} +  \frac{\Stk}{p\nk} - \frac{\Stck}{\nk}\Big]\\
&= \frac{1}{n^2(n-1)}\sum_{k=1}^K(n-\nk)\nk\text{var}(\taukest|\mathcal{S})
\end{align*}

Then we can write the difference as
\begin{align*}
\text{var}\left(\tauestcr|\mathcal{S}\right) &- \text{var}\left(\tauestbk|\mathcal{S}\right) \\
=&\frac{1}{n-1}\left[ \text{var}_k\left(\sqrt{\frac{p}{1-p}}\meanpock + \sqrt{\frac{1-p}{p}}\meanpotk\right)  - \sum_{k=1}^K\frac{(n-\nk)\nk}{n^2}\text{var}(\taukest|\mathcal{S})\right]\\
=&\sum_{k=1}^K\frac{\nk}{n(n-1)}\Big[\left(\sqrt{\frac{p}{1-p}}\meanpock +\sqrt{\frac{1-p}{p}}\meanpotk- \left(\sqrt{\frac{p}{1-p}}\meanpoc+\sqrt{\frac{1-p}{p}}\meanpot\right)\right)^2 \\
& - \frac{n-\nk}{n}\text{var}(\taukest|\mathcal{S})\Big].
\end{align*}

To write this another way, note that
\begin{align*}
\Stck &= \frac{1}{\nk-1}\sum_{i:b_i=k}\left(\tau_i - \taukfs\right)^2\\
&= \frac{1}{\nk-1}\sum_{i:b_i=k}\left[\left(\pot - \meanpotk\right)^2 + \left(\poc - \meanpock\right)^2 -2\left(\pot - \meanpotk\right)\left(\poc - \meanpock\right)\right].
\end{align*}

So then
\begin{align*}
&\frac{\Sck}{1-p} +\frac{\Stk}{p} - \Stck\\
&=\frac{\frac{1}{\nk-1}\sum_{i:b_i=k}\left(\poc - \meanpock\right)^2}{1-p} +\frac{\frac{1}{\nk-1}\sum_{i:b_i=k}\left(\pot - \meanpotk\right)^2}{p}\\
& - \frac{1}{\nk-1}\sum_{i:b_i=k}\left[\left(\pot - \meanpotk\right)^2 + \left(\poc - \meanpock\right)^2 -2\left(\pot - \meanpotk\right)\left(\poc - \meanpock\right)\right]\\
&=\sum_{i:b_i=k}\frac{1}{\nk-1}\left[\frac{p\left(\poc - \meanpock\right)^2}{1-p} +\frac{(1-p)\left(\pot - \meanpotk\right)^2}{p} + 2\left(\pot - \meanpotk\right)\left(\poc - \meanpock\right)\right]\\
&=\sum_{i:b_i=k}\frac{1}{\nk-1}\left(\sqrt{\frac{p}{1-p}}\left(\poc - \meanpock\right) +\sqrt{\frac{1-p}{p}}\left(\pot - \meanpotk\right)\right)^2\\
&=\sum_{i:b_i=k}\frac{1}{\nk-1}\left(\sqrt{\frac{p}{1-p}}\poc + \sqrt{\frac{1-p}{p}}\pot-\left(\sqrt{\frac{p}{1-p}}\meanpock+\sqrt{\frac{1-p}{p}}\meanpotk\right)\right)^2.
\end{align*}

So we get
\begin{align*}
\text{var}\left(\tauestcr|\mathcal{S}\right) &- \text{var}\left(\tauestbk|\mathcal{S}\right) \\
=&\sum_{k=1}^K\frac{1}{n(n-1)}\Big[\nk\left(\sqrt{\frac{p}{1-p}}\meanpock +\sqrt{\frac{1-p}{p}}\meanpotk- \left(\sqrt{\frac{p}{1-p}}\meanpoc+\sqrt{\frac{1-p}{p}}\meanpot\right)\right)^2 \\
& - \frac{n-\nk}{n(\nk-1)}\sum_{i:b_i=k}\left(\sqrt{\frac{p}{1-p}}\poc + \sqrt{\frac{1-p}{p}}\pot-\left(\sqrt{\frac{p}{1-p}}\meanpock+\sqrt{\frac{1-p}{p}}\meanpotk\right)\right)^2\Big].
\end{align*}

This allows us to see the similarity in the two terms.

\subsection{Simple random sampling: Independence case}\label{append:cr_vs_bk_srs}
The simple random sampling framework has two steps: First, obtain a random sample of units for the experiment.
Then, we form blocks.
We assume there is a fixed procedure that clusters on the observed covariate matrix $\bm{X}$ to form blocks for the given sample.
The number and size of blocks can be sample dependent.
However, conditioning on $\bm{X}$ fixes the number and sizes of blocks.

Now think about a ``bad'' case, that this set of covariates, $\bm{X}$, are actually independent of the potential outcomes, i.e. not predictive.
Then we have, assuming equal proportions treated in each block,
\begin{align*}
\text{var}&\left(\tauestcr|\model_{\text{SRS}}\right) - \text{var}\left(\tauestbk|\model_{\text{SRS}}\right)\\
&= \EE\left[\text{var}(\tauestcr|\mathcal{S})|\model_{\text{SRS}}\right]+\text{var}\left(\EE[\tauestcr|\mathcal{S}]|\model_{\text{SRS}}\right) - \EE\left[\text{var}(\tauestbk|\mathcal{S})|\model_{\text{SRS}}\right]-\text{var}\left(\EE[\tauestbk|\mathcal{S}]|\model_{\text{SRS}}\right)\\
&=\EE\left[\frac{\Sc}{n_c}+\frac{\St}{n_t}-\frac{\Stc}{n}\Big|\model_{\text{SRS}}\right] - \EE\left[\sum_{k=1}^K\frac{\nk^2}{n^2}\left(\frac{\Sck}{\nck}+\frac{\Stk}{\ntk}-\frac{\Stck}{\nk}\right)\Big|\model_{\text{SRS}}\right]\\
&=\frac{\sigmac}{n_c}+\frac{\sigmat}{n_t}-\frac{\sigmatc}{n} - \EE\left[\sum_{k=1}^K\frac{\nk}{n}\left(\frac{\Sck}{n_c}+\frac{\Stk}{n_t}-\frac{\Stck}{n}\right)\Big|\model_{\text{SRS}}\right]\\
&= \frac{\sigmac}{n_c}+\frac{\sigmat}{n_t}-\frac{\sigmatc}{n} - \EE\left[\EE\left[\sum_{k=1}^K\frac{\nk}{n}\left(\frac{\Sck}{n_c}+\frac{\Stk}{n_t}-\frac{\Stck}{n}\right)\Big|\model_{\text{SRS}}, \bm{X}\right]\Big|\model_{\text{SRS}}\right]\\
&=\frac{\sigmac}{n_c}+\frac{\sigmat}{n_t}-\frac{\sigmatc}{n} - \EE\left[\sum_{k=1}^K\frac{\nk}{n}\left(\frac{\sigmac}{n_c}+\frac{\sigmat}{n_t}-\frac{\sigmatc}{n}\right)\Big|\model_{\text{SRS}}\right]\\
&=\frac{\sigmac}{n_c}+\frac{\sigmat}{n_t}-\frac{\sigmatc}{n}- \left(\frac{\sigmac}{n_c}+\frac{\sigmat}{n_t}-\frac{\sigmatc}{n}\right)\\
&=0.
\end{align*}

Note that conditioning on $\bm{X}$ fixes the number of units in each block and the total number of blocks, but potential outcomes are independent of this value of the covariates so we can push the expectation through.

Another way to think about independent covariates is that blocking on $\bm{X}$ essentially induces random blocking.
In the finite sample setting, random blocking is the process of randomly partitioning units into a fixed number of $K$ blocks of fixed size $n_k$ for $k=1,\dots K$.
That is, we first would specify the number and sizes of blocks and then randomly assign the appropriate number of units to each block.
After this initial random step, we would then perform a blocked randomized experiment.
As long as in our initial random blocking step the probability of ending up in each block is the same for all units, this is equivalent to a complete randomization with the same number of total treated units.
This can easily be deduced by noticing that all partitions of the assignment space with the same number of treated units are equally likely under this random blocking design, exactly as under a completely randomized design.

Now consider extending this notion of random blocking to the simple random sampling setting.
Because the covariates are independent of outcomes across sampling steps, blocking is essentially done at random across different samples.
Hence, our blocked design will be equivalent to the completely randomized design.

\subsection{Proof of Theorem~\ref{prop:cr_vs_bk_strat_samp}: Variance comparison under stratified sampling}\label{append:cr_vs_bk_strat_samp}

\begin{proof}
First use decomposition of variance and then simplify using results from the derivation in Appendix \ref{append:cr_vs_bk_fs}.
Let $\popmeanpoz$ and $\sigmaz$ be the population mean and variance, respectively, of the potential outcomes of all units under treatment $z$.
Let  $\popmeanpozk$ and $\sigmazk$ be the population mean and variance of the potential outcomes of units in stratum $k$ under treatment $z$. 

\begin{align}
&\text{var}(\tauestcr|\model_{\text{strat}}) - \text{var}(\tauestbk|\model_{\text{strat}})\nonumber\\
=& \EE\left[\text{var}(\tauestcr|\mathcal{S})|\model_{\text{strat}}\right] + \text{var}\left(\EE[\tauestcr|\mathcal{S}]|\model_{\text{strat}}\right)- \EE\left[\text{var}(\tauestbk|\mathcal{S})|\model_{\text{strat}}\right] - \text{var}\left(\EE[\tauestbk|\mathcal{S}]|\model_{\text{strat}}\right)\nonumber\\
=& \EE\left[\text{var}(\tauestcr|\mathcal{S})|\model_{\text{strat}}\right] - \EE\left[\text{var}(\tauestbk|\mathcal{S})|\model_{\text{strat}}\right]\nonumber\\
=& \EE\Big[\underbrace{\frac{\sum_{k=1}^K \nk\left(\meanpock - \meanpoc\right)^2}{(1-p)(n-1)n} + \frac{\sum_{k=1}^K\nk\left(\meanpotk - \meanpot\right)^2}{p(n-1)n}}_{\textbf{A}} - \underbrace{\frac{\sum_{k=1}^K\nk(\taufs - \taukfs)^2}{n(n-1)}}_{\textbf{B}}|\model_{\text{strat}}\Big] \nonumber\\
& - \sum_{k=1}^K\frac{n-\nk}{n^2(n-1)}\left(\frac{\sigmack}{(1-p)} + \frac{\sigmatk}{p} -\sigmatck\right) \label{eq:deriv_cr_vs_bk_strat_samp}
\end{align}

We start with expectation of the numerators in term \textbf{A}:

\begin{align*}
\EE\left[\left(\meanpozk - \meanpoz\right)^2|\model_{\text{strat}}\right] &= \EE\left[\meanpozk^2- 2\meanpozk\meanpoz + \meanpoz^2|\model_{\text{strat}}\right]\\
&= \underbrace{\EE\left[\meanpozk^2|\model_{\text{strat}}\right]}_{\textbf{A.1}}- 2\EE\left[\meanpozk\meanpoz|\model_{\text{strat}}\right] + \underbrace{\EE\left[\meanpoz^2|\model_{\text{strat}}\right]}_{\textbf{A.2}}.
\end{align*}

\textbf{A.1}:
\begin{align*}
\EE\left[\meanpozk^2|\model_{\text{strat}}\right]&= \text{var}\left(\meanpozk|\model_{\text{strat}}\right) + \EE\left[\meanpozk|\model_{\text{strat}}\right]^2= \frac{\sigmazk}{\nk} + \popmeanpozk^2
\end{align*}

\textbf{A.2}:
\begin{align*}
\EE\left[\meanpoz^2|\model_{\text{strat}}\right] &= \text{var}\left(\meanpoz|\model_{\text{strat}}\right) + \EE\left[\meanpoz|\model_{\text{strat}}\right]^2
=\sum_{k=1}^K\frac{\nk}{n^2}\sigmazk + \left(\sum_{k=1}^K\frac{\nk}{n}\popmeanpozk\right)^2
\end{align*}

Putting \textbf{A.1} and \textbf{A.2} together and combining like terms:
\begin{align*}
\sum_{k=1}^k\frac{\nk}{n}\EE\left[\left(\meanpozk - \meanpoz\right)^2|\model_{\text{strat}}\right] 
&= \sum_{k=1}^K\frac{\nk}{n}\EE\left[\meanpozk^2|\model_{\text{strat}}\right] - \EE\left[\meanpoz^2|\model_{\text{strat}}\right]\\
&= \sum_{k=1}^K\frac{n-\nk}{n^2}\sigmazk + \sum_{k=1}^K\frac{\nk}{n}\popmeanpozk^2 - \left(\sum_{k=1}^K\frac{\nk}{n}\popmeanpozk\right)^2.
\end{align*}

The expectation of \textbf{B} follows in a very similar manner, except now we have treated and control units in the calculation.

\textbf{B} becomes
\begin{align*}
&\sum_{k=1}^K\frac{\nk}{n}\EE\left[\left(\taufs - \taukfs\right)^2|\model_{\text{strat}}\right]\\
 &= \sum_{k=1}^K\frac{n-\nk}{n^2}\sigmatck + \sum_{k=1}^K\frac{\nk}{n}\left(\popmeanpotk - \popmeanpock\right)^2- \left[\left(\sum_{k=1}^K\frac{\nk}{n}\popmeanpotk\right) - \left(\sum_{k=1}^K\frac{\nk}{n}\popmeanpock \right)\right]^2.
\end{align*}

Putting \textbf{A} and \textbf{B} into Equation~\ref{eq:deriv_cr_vs_bk_strat_samp} and simplifying:

\begin{align*}
&\text{var}(\tauestcr|\model_{\text{strat}}) - \text{var}(\tauestbk|\model_{\text{strat}})\\
&= \frac{\sum_{k=1}^K\frac{\nk}{n}\popmeanpock^2 - \left(\sum_{k=1}^K\frac{\nk}{n}\popmeanpock\right)^2}{(1-p)(n-1)} + \frac{\sum_{k=1}^K\frac{\nk}{n}\popmeanpotk^2 - \left(\sum_{k=1}^K\frac{\nk}{n}\popmeanpotk\right)^2}{p(n-1)}\\
& - \frac{\sum_{k=1}^K\frac{\nk}{n}\left(\popmeanpotk - \popmeanpock\right)^2 - \left(\sum_{k=1}^K\frac{\nk}{n}\left(\popmeanpotk - \popmeanpock\right)\right)^2}{n-1}\\
&= \frac{p}{1-p}\frac{\sum_{k=1}^K\frac{\nk}{n}\popmeanpock^2 - \left(\sum_{k=1}^K\frac{\nk}{n}\popmeanpock\right)^2}{n-1} + \frac{1-p}{p}\frac{\sum_{k=1}^K\frac{\nk}{n}\popmeanpotk^2 - \left(\sum_{k=1}^K\frac{\nk}{n}\popmeanpotk\right)^2}{n-1}\\
& + 2\frac{\sum_{k=1}^K\frac{\nk}{n}\popmeanpotk\popmeanpock - \left(\sum_{k=1}^K\frac{\nk}{n}\mu(t, k)\right)\left(\sum_{k=1}^K\frac{\nk}{n}\mu(c, k)\right)}{n-1}\\
&=\frac{1}{n-1}\left[\frac{p}{1-p}\text{Var}_k\left(\popmeanpock\right) + \frac{1-p}{p}\text{Var}_k\left(\popmeanpotk\right) + 2Cov_k\left(\popmeanpock, \popmeanpotk\right)\right]\\
&=\frac{1}{n-1}\text{Var}_k\left(\sqrt{\frac{p}{1-p}}\popmeanpotk + \sqrt{\frac{1-p}{p}}\popmeanpock\right)\geq 0.
\end{align*}
The variance in the last line is the variance over the blocks, as defined in Equation \ref{eq:var_k_def}. Therefore we have that $\text{var}(\tauestcr|\model_{\text{strat}}) - \text{var}(\tauestbk|\model_{\text{strat}}) \geq 0$ so we are always doing better with blocking in this setting.
\end{proof}

\subsection{Stratified sampling vs SRS  comparisons}\label{append:cr_srs_vs_strat_samp}
\begin{corollary}\label{cor:strat_samp_vs_srs}
The difference between $\text{var}(\tauestcr|\model_{\text{SRS}}) - \text{var}\left(\tauestcr|\model_{\text{strat}}\right)$ may be positive or negative.
\end{corollary}
\begin{proof}
Compare the previous result to when we assume $\model_{\text{SRS}}$ for complete randomization and $\model_{\text{strat}}$ for blocked randomization. 
\begin{align*}
\text{var}&(\tauestcr|\model_{\text{SRS}}) - \text{var}(\tauestbk|\model_{\text{strat}}) \nonumber\\
 &= \frac{1}{n_{c}}\left(\sum_{k=1}^{K}\frac{n_k}{n}\left(\popmeanpock-\popmeanpoc\right)^2\right) + \frac{1}{n_{t}}\left(\sum_{k=1}^{K}\frac{n_k}{n}\left(\popmeanpotk-\popmeanpot\right)^2\right)\\
& \geq 0 \nonumber
\end{align*}
A form of this result can be found in \cite{imbens2011experimental}.

We can then take the difference, $\text{var}(\tauestcr|\model_{\text{SRS}}) - \text{var}\left(\tauestcr|\model_{\text{strat}}\right)$, between the results under the two frameworks to see whether using different sampling frameworks over or under estimates the benefits of blocking.
First consider the form of the variance under the two different models.
\begin{align*}
\text{var}(\tauestcr|\model_{\text{SRS}}) &= \frac{\sigmac}{n_c} +  \frac{\sigmat}{n_t}\\
&= \frac{\sum_{k=1}^K\frac{\nk}{n}\sigmack + \sum \frac{\nk}{n}\left(\popmeanpock - \popmeanpoc\right)^2}{n_c} +  \frac{\sum_{k=1}^K\frac{\nk}{n}\sigmatk + \sum \frac{\nk}{n}\left(\popmeanpotk - \popmeanpot\right)^2}{n_t}
\end{align*}

\begin{align*}
\text{var}(\tauestcr|\model_{\text{strat}}) &= \EE\left[\text{var}\left(\tauestcr|\mathcal{S}\right)|\model_{\text{strat}}\right] + \text{var}\left(\EE\left[\tauestcr|\mathcal{S}\right]|\model_{\text{strat}}\right) \\
&= \underbrace{\EE\left[\frac{\Sc}{n_c} + \frac{\St}{n_t} - \frac{\Stc}{n}|\model_{\text{strat}}\right]}_{\textbf{A}} + \underbrace{\text{var}\left(\taufs|\model_{\text{strat}}\right)}_{\textbf{B}}
\end{align*}

For \textbf{A} we can use similar techniques from previous proofs to break the pieces up by block:
\begin{align*}
\EE\left[\frac{\Stc}{n}|\model_{\text{strat}}\right]&=\sum_{k=1}^K\left[ \frac{\nk}{n^2}\sigmatck + \frac{\nk}{n(n-1)}\left(\tauk - \tau\right)^2\right]\\
\EE\left[\frac{\Sz}{n_{z}}|\model_{\text{strat}}\right]&=\frac{1}{n_z}\sum_{k=1}^K\left[ \frac{\nk}{n}\sigmazk + \frac{\nk}{n-1}\left(\popmeanpozk - \popmeanpoz\right)^2\right]
\end{align*}

For \textbf{B} we can split up by block and then use classical results from sampling theory on variation of sample means under $\model_{\text{SRS}}$ \cite[Chapter~2]{lohr2010sampling}.

\begin{align*}
\text{var}\left(\taufs|\model_{\text{strat}}\right) &= \sum_{k=1}^K\frac{\nk^2}{n^2}\frac{\sigmatck}{\nk}
\end{align*}

Putting \textbf{A} and \textbf{B} together and collecting similar terms, we have
\begin{align*}
\text{var}(\tauestcr|\model_{\text{strat}}) 
=& \sum_{k=1}^K\left[ \frac{\nk}{n}\frac{\sigmack}{n_c} + \frac{\nk}{n}\frac{\sigmatk}{n_t}\right]\\
& + \sum_{k=1}^K\left[\frac{\nk}{(n-1)n_c}\left(\popmeanpock - \popmeanpoc\right)^2 + \frac{\nk}{(n-1)n_t}\left(\popmeanpotk - \popmeanpot\right)^2\right] \\
& -\sum_{k=1}^K \frac{\nk}{n(n-1)}\left(\tauk - \tau\right)^2.
\end{align*}

So then we can do the subtraction:
\begin{align*}
\text{var}&(\tauestcr|\model_{\text{SRS}}) - \text{var}\left(\tauestcr|\model_{\text{strat}}\right)\\
&=\frac{1}{n(n-1)}\sum_{k=1}^Kn_k\left[\left(\tau_k - \tau\right)^2 -\frac{\left(\popmeanpock -\popmeanpoc\right)^2}{n_c}-\frac{\left(\popmeanpotk -\popmeanpot\right)^2}{n_t}\right]\\
&=\sum_{k=1}^K\frac{\nk}{n(n-1)}\Big[\frac{(n_c-1)\left(\popmeanpock -\popmeanpoc\right)^2}{n_c}+\frac{(n_t-1)\left(\popmeanpotk -\popmeanpot\right)^2}{n_t}\\
& \quad \quad \quad \quad \quad \quad \quad \quad - 2\left(\popmeanpock -\popmeanpoc\right)\left(\popmeanpotk -\popmeanpot\right)\Big].
\end{align*}

We see that this difference could be positive or negative.
In particular, if the blocks all have the same average treatment effect but different control and treatment means, then this expression will be negative, indicating that comparing the variance of blocking under $\model_{\text{strat}}$ to complete randomization under $\model_{\text{SRS}}$ is an underestimate of the benefits of blocking.
On the other hand, if $n_c$ and $n_t$ are large compared to the variation of block average control and treatment outcomes, then we would expect the negative terms in the first expression to be small, resulting in the difference being positive.
This would mean that there is an overestimate of the benefits of blocking when comparing the variance of blocking under $\model_{\text{strat}}$ to complete randomization under $\model_{\text{SRS}}$.

If we have $n_c$ and $n_t$ large so $\frac{n_c-1}{n_c} \approx 1 \approx \frac{n_t-1}{n_t}$ then we have
\begin{align*}
\text{var}\left(\tauestcr|\model_{\text{SRS}}\right) - \text{var}\left(\tauestcr|\model_{\text{strat}}\right) &= \frac{1}{n(n-1)}\sum_{k=1}^K\nk\left(\tauk-\tau\right)^2.
\end{align*}
\end{proof}

\subsection{Unequal treatment proportions}\label{append:cr_vs_bk_unequal}

Here we present the proof for Theorem~\ref{theorem:unequal_prop}.

\begin{proof}
Let the proportion of units treated in a complete randomization be $p$ and in blocked randomization be $p_k$ for block $k$.
Again, we need to break the complete randomization variance into block components.
For the finite sample, the complete randomization variance, from Appendix~\ref{append:cr_vs_bk_fs} is

\begin{align*}
\text{var}\left(\tauestcr|\mathcal{S}\right) 
&= \frac{\sum_{k=1}^K\frac{\nk-1}{n-1}\Sck + \sum_{k=1}^K\frac{\nk}{n-1}\left(\meanpock - \meanpoc\right)^2}{(1-p)n}\\
& + \frac{\sum_{k=1}^K\frac{\nk-1}{n-1}\Stk + \sum_{k=1}^K\frac{\nk}{n-1}\left(\meanpotk - \meanpot\right)^2}{pn}\\
& - \frac{\sum_{k=1}^K\frac{\nk-1}{n-1}\Stck + \sum_{k=1}^K\frac{\nk}{n-1}\left(\taukfs-\taufs\right)^2}{n}.
\end{align*}

For blocked randomization, the variance is
\begin{align*}
\text{var}\left(\tauestbk|\mathcal{S}\right) &= \sum_{k=1}^K\frac{\nk^2}{n^2}\left(\frac{\Sck}{\nck} + \frac{\Stk}{\ntk} - \frac{\Stck}{\nk}\right)\\
&= \sum_{k=1}^K\frac{\nk}{n^2}\left(\frac{\Sck}{1-p_k} + \frac{\Stk}{p_k} - \Stck\right).
\end{align*}

Then the difference is
\begin{align*}
&\text{var}\left(\tauestcr|\mathcal{S}\right)-\text{var}\left(\tauestbk|\mathcal{S}\right)\\
 &=\sum_{k=1}^K\left[\frac{(1-p_k)(\nk-1)n - (1-p)\nk(n-1)}{(1-p_k)(1-p)n^2(n-1)}\Sck+\frac{p_k(\nk-1)n - p\nk(n-1)}{p_kpn^2(n-1)}\Stk\right]\\
 & + \sum_{k=1}^K\left[\frac{n-\nk}{(n-1)n^2}\Stck\right]+\frac{\sum_{k=1}^K\nk\left[\frac{\left(\meanpock - \meanpoc\right)^2}{1-p} + \frac{\left(\meanpotk - \meanpot\right)^2}{p} - \left(\taukfs-\taufs\right)^2\right]}{(n-1)n}\\
 \color{red}
  &=\frac{1}{n^2(n-1)}\sum_{k=1}^K\Bigg(\frac{(1-p_k)(\nk-1)n - (1-p)\nk(n-1)}{(1-p_k)(1-p)}\Sck+\frac{p_k(\nk-1)n - p\nk(n-1)}{p_kp}\Stk\\
  &\quad \quad \quad \quad \quad \quad+ (n-\nk)\Stck\Bigg)+\frac{\sum_{k=1}^K\nk\left[\frac{\left(\meanpock - \meanpoc\right)^2}{1-p} + \frac{\left(\meanpotk - \meanpot\right)^2}{p} - \left(\taukfs-\taufs\right)^2\right]}{(n-1)n}.
\end{align*}
\color{black}

For the stratified random sampling ($\model_{\text{strat}}$), we use results from the proof in Appendix \ref{append:cr_vs_bk_strat_samp} and the above derivation, and combine like terms to get

\begin{align*}
&\text{var}\left(\tauestcr|\model_{\text{strat}}\right)-\text{var}\left(\tauestbk|\model_{\text{strat}}\right)\\
  &=\frac{1}{n-1}\sum_{k=1}^K\frac{\nk}{n}\left(\sqrt{\frac{1-p}{p}}\popmeanpock + \sqrt{\frac{p}{1-p}}\popmeanpotk - \left[\sqrt{\frac{1-p}{p}}\popmeanpoc + \sqrt{\frac{p}{1-p}}\popmeanpot \right]\right)^2\\
&\quad +\sum_{k=1}^K\left[\frac{(p-p_k)n_k}{(1-p_k)(1-p)n^2}\sigmack+\frac{(p_k - p)\nk}{p_kpn^2}\sigmatk\right].
\end{align*}
More details on the derivation can be given upon request.

This can also be written as
\begin{align}
&\text{var}\left(\tauestcr|\model_{\text{strat}}\right)-\text{var}\left(\tauestbk|\model_{\text{strat}}\right) \nonumber\\
  &=\frac{1}{n-1}\text{Var}_k\left(\sqrt{\frac{p}{1-p}}\popmeanpock + \sqrt{\frac{1-p}{p}}\popmeanpotk \right)+\sum_{k=1}^K\frac{(p-p_k)n_k}{n^2}\left[\frac{\sigmack}{(1-p_k)(1-p)}-\frac{\sigmatk}{p_kp}\right] .\label{eq:cr_vs_bk_unequal_prop}
\end{align}

The first term in Equation \ref{eq:cr_vs_bk_unequal_prop} is equal to the result in Proposition \ref{prop:cr_vs_bk_strat_samp} and the second term, capturing the additional variability due to varying proportions, is zero when $p_k=p$ for all $k$.
In general, the second term of Equation \ref{eq:cr_vs_bk_unequal_prop} can be positive or negative.

In particular, if by some bad luck, $pp_k\sigmack > (1-p)(1-p_k)\sigmatk$ for all blocks where $p_k > p$ and $pp_k\sigmack < (1-p)(1-p_k)\sigmatk$ for all blocks where $p_k < p$, this term will be negative.
If, in this case, the population mean potential outcomes for all blocks are approximately equal, the entire expression will be negative.
The two terms are of the same order with respect to sample size $n$, making comparison easier.
\end{proof}

\section{Proofs of consequences of ignoring blocking}\label{append:ignore_block}

\subsection{Proof of Theorem~\ref{theorem:pretend_cr_finite}}
\label{append:use_cr_fin_samp}

\begin{proof}
We perform a blocked randomization but then use the variance estimator from a complete randomization, $\frac{\scest}{n_c} + \frac{\stest}{n_t}$.
We will condition on $\textbf{P}_{blk}$, the assignment mechanism being blocked randomization, throughout to make this clear.
For the finite sample framework, the true variance would still be
\[\text{var}\left(\tauestbk|\mathcal{S}, \textbf{P}_{blk}\right) = \frac{\nk^2}{n^2}\left(\frac{\Sck}{\nck} + \frac{\Stk}{\ntk} - \frac{\Stck}{\nk}\right).\] 

Again, we assume that $p_k = p$ for all $k=1,\hdots,K$.
Then 
\[\meanpozobs = \sum_{k=1}^K \frac{\nzk}{n_z}\meanpozkobs = \sum_{k=1}^K \frac{n_k}{n}\meanpozkobs .\]
We have
\begin{align*}
\szest=& \frac{1}{n_z-1}\sum_{i: Z_i=z} \left(\poz - \meanpozobs\right)^2\\
=& \frac{1}{n_z-1}\sum_{k=1}^K \sum_{i: Z_i=z, b_i = k} \left(\poz -\meanpozk + \meanpozk- \meanpozobs\right)^2\\
=& \frac{1}{n_z-1}\sum_{k=1}^K \sum_{i: Z_i=z, b_i = k}\Big[ \left(\poz -\meanpozk\right)^2 +2 \left(\poz -\meanpozk\right)\left(\meanpozk- \meanpozobs\right)\\
&\qquad \qquad \qquad \qquad \qquad + \left(\meanpozk- \meanpozobs\right)^2\Big]\\
=& \frac{1}{n_z-1}\Big[\underbrace{\sum_{k=1}^K \sum_{i: Z_i=z, b_i = k} \left(\poz -\meanpozk\right)^2}_{\textbf{A}} +2 \underbrace{\sum_{k=1}^K\nzk\left(\meanpozkobs -\meanpozk\right)\left(\meanpozk- \meanpozobs\right)}_{\textbf{B}}\\
&\qquad \qquad + \underbrace{\sum_{k=1}^K\nzk\left(\meanpozk- \meanpozobs\right)^2}_{\textbf{C}}\Big].
\end{align*}

Now expand and take the expectation of each term separately. We start with \textbf{A} and note that $\EE\left[\mathbb{I}_{Z_i=z}|\textbf{P}_{blk}\right]$ is the same for all units because the proportion treated is assumed to be the same in all blocks.
\begin{align*}
\EE\left[\sum_{k=1}^K \sum_{i: Z_i=z, b_i = k} \left(\poz -\meanpozk\right)^2|\mathcal{S}, \textbf{P}_{blk}\right]&=\EE\left[\sum_{k=1}^K \sum_{i: b_i = k}\mathbb{I}_{Z_i=z}\left(\poz -\meanpozk\right)^2|\mathcal{S},\textbf{P}_{blk}\right]\\
&=\sum_{k=1}^K \EE\left[\mathbb{I}_{Z_i=z}|\textbf{P}_{blk}\right](\nk-1)\Szk
\end{align*}

Now \textbf{B},

\begin{align*}
&\EE\left[\sum_{k=1}^K\nzk\left(\meanpozkobs -\meanpozk\right)\left(\meanpozk- \meanpozobs\right)|\mathcal{S}, \textbf{P}_{blk}\right]\\
&=\EE\left[\sum_{k=1}^K\nzk\left(\meanpozkobs\meanpozk -\meanpozk^2+\meanpozk\meanpozobs- \meanpozkobs\meanpozobs\right)|\mathcal{S}, \textbf{P}_{blk}\right]\\
&=\EE\left[\sum_{k=1}^K\nzk\left(\meanpozk\meanpozobs- \meanpozkobs\meanpozobs\right)|\mathcal{S}, \textbf{P}_{blk}\right]\\
&=n_z\left(\meanpoz^2- \EE\left[\meanpozobs^2|\mathcal{S}, \textbf{P}_{blk}\right]\right)\\
&=n_z\left(- \text{var}\left(\meanpozobs|\mathcal{S}, \textbf{P}_{blk}\right)\right)\\
&=-n_z\sum_{k=1}^K\frac{\nk^2}{n^2}\text{var}\left(\meanpozkobs|\mathcal{S}, \textbf{P}_{blk}\right)\\
&=-n_z\sum_{k=1}^K\frac{\nzk^2}{n_z^2}\frac{\nk - \nzk}{\nk}\frac{\Szk}{\nzk}\\
&=-\sum_{k=1}^K\frac{\nzk}{n_z}\left(1-\EE\left[\mathbb{I}_{Z_i=z}|\textbf{P}_{blk}, b_i=k\right]\right)\Szk.
\end{align*}

Now \textbf{C},

\begin{align*}
\EE&\left[\sum_{k=1}^K\nzk\left(\meanpozk- \meanpozobs\right)^2|\mathcal{S}\right]\\
 &=\sum_{k=1}^K\nzk\meanpozk^2- 2n_z\meanpoz^2+n_z\EE\left[\meanpozobs^2|\mathcal{S}, \textbf{P}_{blk}\right]\\
&=\sum_{k=1}^K\nzk\meanpozk^2- 2n_z\meanpoz^2+n_z\text{var}\left(\meanpozobs|\mathcal{S}, \textbf{P}_{blk}\right) + n_z\EE\left[\meanpozobs|\mathcal{S}, \textbf{P}_{blk}\right]^2\\
&=\sum_{k=1}^K\nzk\meanpozk^2- n_z\meanpoz^2+\sum_{k=1}^K\frac{\nzk(1-\EE\left[\mathbb{I}_{Z_i=z}|\textbf{P}_{blk}\right])}{n_z}\Szk .
\end{align*}

Putting it all back together,

\begin{align*}
\EE&\left[\szest|\mathcal{S}, \textbf{P}_{blk}\right]\\
=&\frac{1}{n_z-1}\Big[\sum_{k=1}^K \EE\left[\mathbb{I}_{Z_i=z}\right](\nk-1)\Szk -2\sum_{k=1}^K\frac{\nzk}{n_z}\left(1-\EE\left[\mathbb{I}_{Z_i=z|\textbf{P}_{blk}}\right]\right)\Szk+ \sum_{k=1}^K\nzk\meanpozk^2- n_z\meanpoz^2\\
&+\sum_{k=1}^K\frac{\nzk(1-\EE\left[\mathbb{I}_{Z_i=z}|\textbf{P}_{blk}\right])}{n_z}\Szk\Big]\\
=&\frac{1}{n_z-1}\Big[\sum_{k=1}^K \EE\left[\mathbb{I}_{Z_i=z}|\textbf{P}_{blk}\right](\nk-1)\Szk -\sum_{k=1}^K\frac{\nzk}{n_z}\left(1-\EE\left[\mathbb{I}_{Z_i=z}|\textbf{P}_{blk}\right]\right)\Szk\\
&+ \sum_{k=1}^K\nzk\meanpozk^2-n_z\meanpoz^2\Big]\\
=&\sum_{k=1}^K\left(\frac{\nk}{n}-\frac{\EE\left[\mathbb{I}_{Z_i=z}|\textbf{P}_{blk}\right](n-\nk)}{n(n_z-1)}\right)\Szk+ \frac{1}{n_z-1}\sum_{k=1}^K\nzk\left(\meanpozk- \meanpoz\right)^2.
\end{align*}

So then the bias is

\begin{align*}
&\EE\left[\frac{\scest}{n_c} + \frac{\stest}{n_t}|\mathcal{S}, \textbf{P}_{blk}\right] - \sum_{k=1}^K\frac{\nk^2}{n^2}\left(\frac{\Sck}{\nck} + \frac{\Stk}{\ntk} - \frac{\Stck}{\nk}\right)\\
&=\sum_{k=1}^K\left(\frac{\nk}{(1-p)n^2}-\frac{(1-p)(n-\nk)}{(1-p)n^2(n_c-1)}\right)\Sck+ \frac{1}{n_c-1}\sum_{k=1}^K\frac{\nck}{n_c}\left(\meanpock- \meanpoc\right)^2\\
&+\sum_{k=1}^K\left(\frac{\nk}{pn^2}-\frac{p(n-\nk)}{pn^2(n_t-1)}\right)\Stk+ \frac{1}{n_t-1}\sum_{k=1}^K\frac{\ntk}{n_t}\left(\meanpotk- \meanpot\right)^2\\
& - \sum_{k=1}^K\frac{\nk^2}{n^2}\left(\frac{\Sck}{(1-p)\nk} + \frac{\Stk}{p\nk} - \frac{\Stck}{\nk}\right)\\
&=\frac{1}{n_c-1}\sum_{k=1}^K\frac{\nk}{n}\left(\meanpock- \meanpoc\right)^2+ \frac{1}{n_t-1}\sum_{k=1}^K\frac{\nk}{n}\left(\meanpotk- \meanpot\right)^2\\
&-\left(\sum_{k=1}^K\frac{n-\nk}{n^2(n_c-1)}\Sck+\sum_{k=1}^K\frac{n-\nk}{n^2(n_t-1)}\Stk-\sum_{k=1}^K\frac{\nk}{n^2}\Stck\right).
\end{align*}

If there is no variation between the blocks (i.e. $\meanpozk = \meanpoz$ for all $k$) but there is within block variability (i.e. $\Szk \neq 0$) then this difference will be negative.
The reduction is blunted, however, by the degree of treatment variation there is within block (making $\Stck$ offset the negative term from the $\Szk$).
\end{proof}

\subsection{Proof of Corollary~\ref{cor:pretend_cr_strat}}\label{append:use_cr_strat_samp}

\begin{proof}
We start from the result of Appendix \ref{append:use_cr_fin_samp} and utilize work done in Appendix \ref{append:cr_vs_bk_strat_samp} to simplify things.

\begin{align*}
\EE&\left[\EE\left[\frac{\szest}{n_z}|\mathcal{S}, \textbf{P}_{blk}\right]|\model_{\text{strat}}\right]\\
=&\EE\left[\sum_{k=1}^K\left(\frac{\nk}{nn_z}-\frac{\EE\left[\mathbb{I}_{Z_i=z}|\textbf{P}_{blk}\right](n-\nk)}{n_zn(n_z-1)}\right)\Szk+ \frac{1}{n_z-1}\sum_{k=1}^K\frac{\nzk}{n_z}\left(\meanpozk- \meanpoz\right)^2|\model_{\text{strat}}\right]\\
=&\sum_{k=1}^K\left(\frac{\nk}{nn_z}-\frac{\EE\left[\mathbb{I}_{Z_i=z}|\textbf{P}_{blk}\right](n-\nk)}{n_zn(n_z-1)}\right)\sigmazk+ \EE\left[\frac{1}{n_z-1}\sum_{k=1}^K\frac{\nzk}{n_z}\left(\meanpozk- \meanpoz\right)^2| \model_{\text{strat}}\right]\\
=&\sum_{k=1}^K\left(\frac{\nk}{nn_z}-\frac{n-\nk}{n^2(n_z-1)}\right)\sigmazk+ \frac{1}{n_z-1}\sum_{k=1}^K\frac{n - \nk}{n^2}\sigmazk+\frac{1}{n_z-1}\sum_{k=1}^K\frac{\nk}{n}\left(\popmeanpozk- \popmeanpoz\right)^2\\
=&\sum_{k=1}^K\frac{\nk}{nn_z}\sigmazk+\frac{1}{n_z-1}\sum_{k=1}^K\frac{\nk}{n}\left(\popmeanpozk- \popmeanpoz\right)^2\\
\end{align*}

Now to get the bias we have

\begin{align*}
\EE&\left[\frac{\scest}{n_c} + \frac{\stest}{n_t}|\textbf{P}_{blk}, \model_{\text{strat}}\right] - \sum_{k=1}^K\frac{\nk^2}{n^2}\left(\frac{\sigmack}{\nck} + \frac{\sigmatk}{\ntk}\right)\\
=&\sum_{k=1}^K\frac{\nk}{nn_c}\sigmack+\frac{1}{n_c-1}\sum_{k=1}^K\frac{\nk}{n}\left(\popmeanpock- \popmeanpoc\right)^2 + \sum_{k=1}^K\frac{\nk}{nn_t}\sigmatk+\frac{1}{n_t-1}\sum_{k=1}^K\frac{\nk}{n}\left(\popmeanpotk- \popmeanpot\right)^2\\
&- \sum_{k=1}^K\frac{\nk^2}{n^2}\left(\frac{\sigmack}{\nck} + \frac{\sigmatk}{\ntk}\right)\\
=&\frac{1}{n_c-1}\sum_{k=1}^K\frac{\nk}{n}\left(\popmeanpock- \popmeanpoc\right)^2 +\frac{1}{n_t-1}\sum_{k=1}^K\frac{\nk}{n}\left(\popmeanpotk- \popmeanpot\right)^2.
\end{align*}
\end{proof}


\end{appendices}

\end{document}